\documentclass{article}[12pt]
\usepackage{graphicx,epsfig,color}
 \usepackage{citesort}
\usepackage{enumerate}

\newcounter{savenumi}

\newtheorem{theoremfoo}{Theorem}
\newenvironment{theorem}{\pagebreak[1]\begin{theoremfoo}}{\end{theoremfoo}}

\newtheorem{propositionfoo}[theoremfoo]{Proposition}
\newenvironment{proposition}{\pagebreak[1]\begin{propositionfoo}}{\end{propositionfoo}}
\newtheorem{lemmafoo}[theoremfoo]{Lemma}
\newenvironment{lemma}{\pagebreak[1]\begin{lemmafoo}}{\end{lemmafoo}}
\newtheorem{conjecturefoo}[theoremfoo]{Conjecture}

\newtheorem{corollaryfoo}[theoremfoo]{Corollary}
\newenvironment{corollary}{\pagebreak[1]\begin{corollaryfoo}}{\end{corollaryfoo}}
\newtheorem{exercisefoo}{Exercise}

\newtheorem{openfoo}[theoremfoo]{Question}

\newtheorem{nttn}[theoremfoo]{Notation}

\newtheorem{dfntn}[theoremfoo]{Definition}
\newenvironment{definition}{\pagebreak[1]\begin{dfntn}\rm}{\end{dfntn}}

\newenvironment{proof}
    {\pagebreak[1]{\narrower\noindent {\bf Proof:\quad\nopagebreak}}}{\QED}

\newcommand{\floor}[1]{\left\lfloor#1\right\rfloor}
\newcommand{\ceiling}[1]{\left\lceil#1\right\rceil}

\newcommand{\poly}{{\rm poly}}
\def\nre.{$n$\/-r.e.}
\newcommand{\scrod}{\quad\nopagebreak}

\newtheorem{factfoo}[theoremfoo]{Fact}

\newcommand{\squeeze}{
\textwidth 6in \textheight 8.8in \oddsidemargin 0.2in \topmargin
-0.4in }

\newtheorem{propertyfoo}[theoremfoo]{Property}

\makeatletter

\def\@makechapterhead#1{ \vspace*{50pt} { \parindent 0pt \raggedright
 \ifnum \c@secnumdepth >\m@ne \huge\bf \@chapapp{} \thechapter. \par
 \vskip 20pt \fi \Huge \bf #1\par
 \nobreak \vskip 40pt } }

\def\@sect#1#2#3#4#5#6[#7]#8{\ifnum #2>\c@secnumdepth
     \def\@svsec{}\else
     \refstepcounter{#1}\edef\@svsec{\csname the#1\endcsname.\hskip 1em }\fi
     \@tempskipa #5\relax
      \ifdim \@tempskipa>\z@
        \begingroup #6\relax
          \@hangfrom{\hskip #3\relax\@svsec}{\interlinepenalty \@M #8\par}
        \endgroup
       \csname #1mark\endcsname{#7}\addcontentsline
         {toc}{#1}{\ifnum #2>\c@secnumdepth \else
                      \protect\numberline{\csname the#1\endcsname}\fi
                    #7}\else
        \def\@svsechd{#6\hskip #3\@svsec #8\csname #1mark\endcsname
                      {#7}\addcontentsline
                           {toc}{#1}{\ifnum #2>\c@secnumdepth \else
                             \protect\numberline{\csname the#1\endcsname}\fi
                       #7}}\fi
     \@xsect{#5}}

\def\@begintheorem#1#2{\it \trivlist \item[\hskip \labelsep{\bf #1\ #2.}]}

\def\@opargbegintheorem#1#2#3{\it \trivlist
      \item[\hskip \labelsep{\bf #1\ #2\ (#3).}]}

\makeatother

\newif\ifshortconferences
\shortconferencesfalse
\newif\ifmediumconferences
\mediumconferencesfalse

\def\ending#1{{\count1=#1\relax
\count2=\count1 \divide\count2 by 100 \multiply\count2 by 100
\advance\count1 by -\count2
\ifnum\count1=11
th%
\else \ifnum\count1=12
th%
\else \ifnum\count1=13
th%
\else \count2=\count1 \divide\count1 by 10 \multiply\count1 by 10
\advance\count2 by -\count1 \ifnum\count2=1
st%
\else \ifnum\count2=2
nd%
\else \ifnum\count2=3
rd%
\else th%
\fi\fi\fi\fi\fi\fi }}

\def\Proceedingsofthe{\ifshortconferences Proc.\else\ifmediumconferences Proc.\else Proceedings of the\fi\fi}

\newcounter{confnum}

\def\conf#1#2{%
\setcounter{confnum}{#2}%
\addtocounter{confnum}{-\csname #1zero\endcsname}%
\ifnum\value{confnum}=1%
\expandafter\ifx\csname #1One\endcsname\relax%
\Proceedingsofthe\ \arabic{confnum}\ending{\value{confnum}}\ \csname #1name\endcsname%
\else \csname #1One\endcsname\fi%
\else%
\Proceedingsofthe\ \arabic{confnum}\ending{\value{confnum}}\ \csname
#1name\endcsname\fi}

\def\qsym{\vrule width0.7ex height0.9em depth0ex}

\newif\ifqed\qedtrue

\def\noqed{\global\qedfalse}

\def\qed{\ifqed{\penalty1000\unskip\nobreak\hfil\penalty50
\hskip2em\hbox{}\nobreak\hfil\qsym
\parfillskip=0pt \finalhyphendemerits=0\par\medskip}\fi\global\qedtrue}

\makeatletter
\def\eqnqed{\noqed
    \def\@tempa{equation}
    \ifx\@tempa\@currenvir\def\@eqnnum{\qsym}%
    \addtocounter{equation}{-1}\else%
    \def\@@eqncr{\let\@tempa\relax
    \ifcase\@eqcnt \def\@tempa{& & &}\or \def\@tempa{& &}%
      \else \def\@tempa{&}\fi
     \@tempa {\def\@eqnnum{{\qsym}}\@eqnnum}
     \global\@eqnswtrue\global\@eqcnt\z@\cr}\fi}

\def\eqnlabel#1#2{\if@filesw {\let\thepage\relax%
   \def\protect{\noexpand\noexpand\noexpand}%
   \edef\@tempa{\write\@auxout{\string
      \newlabel{#2}{{{#1}}{\thepage}}}}%
   \expandafter}\@tempa%
   \if@nobreak \ifvmode\nobreak\fi\fi\fi%
    \def\@tempa{equation}
    \ifx\@tempa\@currenvir\def\theequation{{#1}}%
    \addtocounter{equation}{-1}\else%
    \def\@@eqncr{\let\@tempa\relax
    \ifcase\@eqcnt \def\@tempa{& & &}\or \def\@tempa{& &}%
      \else \def\@tempa{&}\fi
     \@tempa {\def\@eqnnum{{#1}}\@eqnnum}
     \global\@eqnswtrue\global\@eqcnt\z@\cr}\fi}

\makeatother

\def\QED{\qed}

\makeatother

\squeeze

\newcommand{\RegularVersion}{}

\newcommand{\littleo}{{\rm o}}
\newcommand{\bigO}{{\rm O}}
\newcommand{\prob}{{\rm Prob}}

\newcommand{\randomElm}{{\rm RandomElement}}
\newcommand{\randomNum}{{\rm RandomNumber}}

\newcommand{\test}{{\rm test}}

\newcommand{\integerSet}{{\rm N}}

\newcommand{\realSet}{{\rm R}}

\newcommand{\setCount}{m} 

\newcommand{\setSize}{n} 

\newcommand{\hashingSize}{H}

\newcommand{\sampleCount}{ s}

\newcommand{\eventCount}{ t}

\newcounter{problem-counter}
\begin{document}

\ifdefined\RegularVersion

\else \setcounter{page}{0}

\fi

\date{}

\title{
Partial Sublinear Time Approximation and Inapproximation for Maximum
Coverage
\thanks{This research was supported in part by National Science
Foundation Early Career Award 0845376 and Bensten Fellowship of the
University of Texas - Rio Grande Valley.} }

\author{
Bin Fu
\\\\
Department of Computer Science\\
 University of Texas - Rio Grande Valley, Edinburg, TX 78539, USA\\
bin.fu@utrgv.edu
\\
\\
} \maketitle

\begin{abstract} We develop a randomized approximation algorithm for
the classical maximum coverage problem, which given a list of sets
$A_1,A_2,\cdots, A_{\setCount}$ and integer parameter $k$,  select
$k$ sets $A_{i_1}, A_{i_2},\cdots, A_{i_k}$ for maximum union
$A_{i_1}\cup A_{i_2}\cup\cdots\cup A_{i_k}$. In our algorithm, each
input set $A_i$ is a black box that can provide its size $|A_i|$,
generate a random element of $A_i$, and answer the membership query
$(x\in A_i?)$ in $O(1)$ time. Our algorithm gives $(1-{1\over
e})$-approximation for maximum coverage problem in
$O(\poly(k)\setCount\cdot\log \setCount)$ time, which is independent
of the sizes of the input sets. No existing
$\bigO(p(\setCount)\setSize^{1-\epsilon})$ time $(1-{1\over
e})$-approximation algorithm for the maximum coverage has been found
for any function $p(\setCount)$ that only depends on the number of
sets, where $\setSize=\max(|A_1|,\cdots,| A_{\setCount}|)$ (the
largest size of input sets). The notion of partial sublinear time
algorithm is introduced. For a computational problem with input size
controlled by two parameters $\setSize$ and $\setCount$, a partial
sublinear time algorithm for it runs in a
$O(p(\setCount)\setSize^{1-\epsilon})$ time or
$O(q(\setSize)\setCount^{1-\epsilon})$ time. The maximum coverage
has a partial sublinear time $O(\poly(\setCount))$ constant factor
approximation since $k\le \setCount$. On the other hand, we show
that the maximum coverage problem has no partial sublinear
$O(q(\setSize)\setCount^{1-\epsilon})$ time constant factor
approximation algorithm. This separates the partial sublinear time
computation from the conventional sublinear time computation by
disproving the existence of sublinear time approximation algorithm
for the maximum coverage problem.
\end{abstract}

\centerline{{\bf Key words:} Maximum Coverage, Greedy Method,
Approximation, Partial Sublinear Time.}

\ifdefined\RegularVersion
\else
\newpage
\fi

\section{Introduction}

The maximum coverage problem is a classical NP-hard problem with
many applications~\cite{CornuejolsFisherNemhauser77,HochbaumBook97},
and is directly related to set cover problem, one of Karp's
twenty-one NP-complete problems~\cite{Karp:NP-complete}. The input
has several sets and a number $k$. The sets may have some elements
in common. You must select at most $k$ of these sets such that the
maximum number of elements are covered, i.e. the union of the
selected sets has a maximum  size. The greedy algorithm for maximum
coverage chooses sets according to one rule: at each stage, choose a
set which contains the largest number of uncovered elements. It can
be shown that this algorithm achieves an approximation ratio of
$(1-{1\over
e})$~\cite{CornuejolsFisherNemhauser77,HochbaumPathria98}.
Inapproximability results show that the greedy algorithm is
essentially the best-possible polynomial time approximation
algorithm for maximum coverage~\cite{Feige98}. The existing
implementation for the greedy $(1-{1\over e})$-approximation
algorithm for the maximum coverage problem needs $\Omega(\setCount
\setSize)$ time for a list of $\setCount$  sets $A_1,\cdots,
A_{\setCount}$ with
$\setSize=|A_1|=|A_2|=\cdots=|A_{\setCount}|$~\cite{HochbaumPathria98,Vazirani01}.
We have not found any existing
$\bigO(p(\setCount)\setSize^{1-\epsilon})$ time algorithm for the
same ratio $(1-{1\over e})$ of approximation for any function
$p(\setCount)$ that only depends on the number of sets. The variant
versions and methods for this problem have been studied in a series
of
papers~\cite{KhullerMossNaor99,AgeevSviridenko04,ChekuriKumar04,CohenKatzir08,Srinivasan01}.


This paper sticks to the original definition of the maximum coverage
problem, and studies its complexity under several concrete models.
In the first model, each set is accessed as a black box that only
provides random elements and answers membership queries. When
$\setCount$ input sets $A_1,A_2,\cdots, A_{\setCount}$ are given,
our model allows random sampling from each of them, and the
cardinality $|A_i|$ (or approximation for $|A_i|$) of each $A_i$ is
also part of the input. The results of the first model can be
transformed into other conventional models. A set could be a set of
points in a geometric shape. For example, a set may be all lattice
points in a $d$-dimensional rectangular shape. If the center
position, and dimension parameters of the rectangle are given, we
can count the number of lattice points and provide a random sample
for them.

A more generalized maximum coverage problem was studied under the
model of submodular set function subject to a matroid
constraint~\cite{NemhauserWolseyFisher77,CalinescuChekuriPalVondrak11,FilmusWard12},
and has same approximation ratio $1-{1\over e}$.
 The maximum coverage
problem in the matroid model has time complexity
$\bigO(r^3\setCount^2\setSize)$~\cite{FilmusWard12}, and
$\bigO(r^2\setCount^3\setSize+\setCount^7)$~\cite{CalinescuChekuriPalVondrak11},
respectively, according to the analysis in~\cite{FilmusWard12},
 where $r$ is the rank of matroid, $\setCount$ is the number of
sets, and $\setSize$ is the size of the largest set.  The maximum
coverage problem in the matroid model has the oracle query to the
submodular function~\cite{CalinescuChekuriPalVondrak11} and is
counted $\bigO(1)$ time per query. Computing the size of union of
input sets is $\#$P-hard if each input set as a black box is a set
of high dimensional rectangular lattice points since \#DNF is
\#P-hard~\cite{ValiantSharp}. Thus, the generalization of submodular
function in the matroid model does not characterize the
computational complexity for these types of problems.
Our model can be applied to this high dimensional space maximum
coverage problem.


In this paper, we develop a randomized algorithm to approximate the
maximum coverage problem. We show an approximation algorithm for
maximum coverage problem with $(1-{1\over e})$-ratio.
For an input list $L$ of finite sets $A_1,\cdots, A_{\setCount}$, an
integer $k$, and parameter $\epsilon\in (0,1)$, our randomized
algorithm returns an integer $z$ and a subset $H\subseteq
\{1,2,\cdots, \setCount\}$ such that $|\cup_{j\in H}A_j|\ge
(1-{1\over e})C^*(L,k)$ and $|H|=k$, and $(1-\epsilon)|\cup_{j\in
H}A_j|\le z\le (1+\epsilon)|\cup_{j\in H}A_j|$, where $C^*(L,k)$ is
the maximum union size for a solution of maximum coverage. Its
complexity is $\bigO({k^6\over \epsilon^2}(\log ({3 \setCount\over
k})
)\setCount)$ and its probability to fail is less than ${1\over 4}$.

Our computational time is independent of the size of each set if the
membership checking for each input set takes one step. When each set
$A_i$ is already saved in an efficient data structure such as
B-tree, we can also provide an efficient random sample, and make a
membership query to each $A_i$ in a $\bigO(\log |A_i|)$ time. This
model also has practical importance because B-tree is often used to
collect a large set of data.  Our algorithms are suitable to
estimate the maximum coverage when there are multiple big data sets,
and each data set is stored in a efficient data structure that can
support efficient random sampling and membership query. The widely
used B-tree in modern data base clearly fits our algorithm. Our
model and algorithm are suitable to support online computation.


We apply the randomized algorithm to several versions of maximum
coverage problem: 1. Each set contains the lattice points in a
rectangle of $d$-dimensional space. It takes $\bigO(d)$ time for a
random element, or membership query. This gives an application to a
\#P-hard problem. 2. Each set is stored in a unsorted array. It
takes $\bigO(1)$ time for a random element, and $\bigO(n)$ time for
membership query. It takes $\bigO(\log n)$ time for a random
element, or membership query. 3. Each set is stored in a sorted
array. 4. Each set is stored in a B-tree. It takes $\bigO(\log n)$
time for a random element, or a membership query. Furthermore,
B-tree can support online version of maximum coverage that has
dynamic input. 5. Each set is stored in a hashing table. The time
for membership query needs some assumption about the performance of
hashing function. We show how the computational time of the
randomized algorithm for maximum coverage depends on the these data
structures.

Sublinear time algorithms have been found for many computational
problems, such as
checking polygon intersections ~\cite{ChazelleLiuMagen05},
estimating the cost of a minimum spanning
tree~\cite{ChazelleRubinfeldTrevisan06,CzumajErgun05,CzumajSohler04-journal},
finding geometric separators~\cite{FuChen06},
property testing~\cite{GoldreichGoldwasserRon98,GoldreichRon00},
etc.

The notion of partial sublinear time computation is introduced in
this paper. It characterizes a class of computational problems that
are sublinear in one of the input parameters, but not necessarily
the other ones. For a function $f(.)$ that maps a list of sets to
nonnegative integers, a $\bigO(p(\setCount)\setSize^{1-\epsilon})$
time or $\bigO(q(\setSize)\setCount^{1-\epsilon})$ time
approximation to $f(.)$ is a partial sublinear time computation. The
maximum coverage has a partial sublinear time constant factor
approximation scheme. We prove that the special case of maximum
coverage problem with equal size of sets, called {\it equal size
maximum coverage}, is as hard as the general case. On the other
hand, we show that the equal size maximum coverage problem has no
partial sublinear $\bigO(q(\setSize)\setCount^{1-\epsilon})$
constant factor approximation randomized algorithm in a randomized
model. Thus, the partial sublinear time computation is separated
from the conventional sublinear time computation via the maximum
coverage problem.

The paper is organized as follows: In Section~\ref{model-sec}, we
define our model of computation and complexity. In
Section~\ref{overview-sec}, we give an overview of our method for
approximating maximum coverage problem. In Section~\ref{union-sec},
we give randomized greedy approximation for the maximum coverage
problem. In Section~\ref{improved-sect}, a faster algorithm is
presented with one round random sampling, which is different from
the multiple rounds random sampling used in Section~\ref{union-sec}.
In Section~\ref{partial-sublinear-sec} , we introduce the notion of
partial sublinear time computation, and prove inapproximability for
maximum coverage if the time is
$\bigO(q(\setSize)\setCount^{1-\epsilon})$. In
Section~\ref{equali-size-section}, we show a special case of maximum
coverage problem that all input sets have the same size, and prove
that it is as hard as the general case. In
Section~\ref{concrete-models-section}, the algorithm is implemented
in more concrete data model for the maximum coverage problem. An
input set can be stored in a sorted array, unsorted array, B-tree,
or hashing function. A set may be represented by a small set of
parameters if it is a set of high dimensional points such as a set
of lattice points in a rectangle shape.

\section{Computational Model and Complexity}\label{model-sec}

In this section, we show our model of computation, and the
definition of complexity. Assume that  $A_1$ and $A_2$ are two sets.
Define $A_2-A_1$ to be the set of elements in $A_2$, but not
in $A_1$.
For a finite set $A$, we use $|A|$, {\it cardinality} of $A$, to be
the number of distinct elements in $A$. For a real number $x$, let
$\ceiling{x}$ be the least integer $y$ greater than or equal to $x$,
and $\floor{x}$ be the largest integer $z$ less than or equal to
$x$.
Let $\integerSet=\{0,1,2,\cdots\}$ be the set of nonnegative
integers, $\realSet=(-\infty,+\infty)$ be the set of all real
numbers,  and $\realSet^+=[0,+\infty)$ be the set of all nonnegative
real numbers.
An integer $s$ is a $(1+\epsilon)$-approximation for $|A|$ if
$(1-\epsilon)|A|\le s\le (1+\epsilon)|A|$.

\begin{definition}\label{input-list-def} The {\it type 0 model} of randomized computation for our algorithm is defined below:
An input $L$ is a list of sets $A_1,A_2,\cdots, A_{\setCount}$
 that provide the cardinality of $A_i$ is  $\setSize_i=|A_i|$ for $i=1,2,\cdots,
\setCount$,
the largest cardinality of input sets
$\setSize=\max\{\setSize_i:1\le i\le \setCount\}$, and
support the following operations:
\begin{enumerate}[1.]
\item
Function \randomElm$(A_i)$ returns a random element $x$ from $A_i$
for $i=1,2,\cdots, \setCount$.
\item
Function Query$(x, A_i$) returns $1$ if $x\in A_i$, and $0$
otherwise.
\end{enumerate}
\end{definition}

\begin{definition}Let parameters $\alpha_L$ and $\alpha_R$ be in
$[0,1)$.
An {\it $(\alpha_L,\alpha_R)$-biased generator} \randomElm$(A)$ for
set $A$ generates an element in $A$ such that for each $y\in A$,
$(1-\alpha_L)\cdot {1\over |A|}\le \prob(\randomElm(A)=y)\le
(1+\alpha_R)\cdot {1\over |A|}$.
\end{definition}

Definition~\ref{input-list-def1} gives the type 1 model, which is a
generalization of type 0 model. It is suitable to apply our
algorithm for high dimensional problems that may not give uniform
random sampling or exact set size. For example, it is not trivial to
count the number of lattice points or generate a random lattice
point in a $d$-dimensional ball with its center not at a lattice
point.

\begin{definition}\label{input-list-def1} The {\it type 1 model} of randomized computation for our algorithm is defined
below: Let real parameters $\alpha_L,\alpha_R,\delta_L, \delta_R$ be
in $[0,1)$. An input $L$ is a list of sets $A_1,A_2,\cdots,
A_{\setCount}$ that provide an approximate cardinality $s_i$ of
$A_i$ with $(1-\delta_L)|A_i|\le s_i\le (1+\delta_R)|A_i|$ for
$i=1,2,\cdots, \setCount$,
the largest approximate cardinality of input sets $s=\max\{s_i:1\le
i\le \setCount\}$, and
support the following operations:
\begin{enumerate}[1.]
\item
Function \randomElm$(A_i)$ is a $(\alpha_L,\alpha_R)$-biased random
generator for $A_i$ for $i=1,2,\cdots, \setCount$.
\item
Function Query$(x, A_i$) returns $1$ if $x\in A_i$, and $0$
otherwise.
\end{enumerate}
\end{definition}


The main problem, which is called maximum coverage, is that given a
list of sets $A_1,\cdots, A_{\setCount}$ and an integer $k$, find
$k$ sets from $A_1, A_2, \cdots, A_{\setCount}$ to maximize the size
of the union of the selected sets in the computational model defined
in Definition~\ref{input-list-def} or
Definition~\ref{input-list-def1}. For real number $a\in [0,1]$, an
approximation algorithm is a $(1-a)$-approximation for the maximum
coverage problem that has input of integer parameter $k$ and a list
of sets $A_1,\cdots, A_{\setCount}$ if it outputs a sublist of sets
$A_{i_1}, A_{i_2},\cdots, A_{i_k}$ such that $|A_{i_1}\cup
A_{i_2}\cup \cdots\cup A_{i_k}|\ge (1-a) |A_{j_1}\cup A_{j_2}\cup
\cdots\cup A_{j_k}|$, where $A_{j_1}, A_{j_2}, \cdots, A_{j_k}$ is a
solution with maximum size of union.

We use the triple $(T(.), R(.), Q(.))$ to characterize the
computational complexity, where
\begin{itemize}
\item
$T(.)$ is a function for the number of steps that each access to
\randomElm$(.)$ or Query(.) is counted one step,
\item
$R(.)$ is a function to count the number of random samples from
$A_i$ for $i=1, 2,\cdots, \setCount$. It is measured by the total
number of times to access those functions \randomElm$(A_i)$ for all
input sets $A_i$, and
\item
$Q(.)$ is a function to count the number of queries to $A_i$ for
$i=1, \cdots, A_{\setCount}$. It is measured by the total number of
times to access those functions Query$(x, A_i$) for all input sets
$A_i$.
\end{itemize}
The  parameters $\epsilon, \gamma,k,\setSize, \setCount$  can be
used to determine the three complexity functions,  where
$\setSize=\max(|A_1|,\cdots,| A_{\setCount}|)$ (the largest
cardinality of input sets), $\epsilon$ controls the accuracy of
approximation, and $\gamma$ controls the failure probability of a
randomized algorithm. Their types could be written as $T(\epsilon,
\gamma,k,\setCount), R(\epsilon, \gamma,k,\setCount)$, and
$Q(\epsilon, \gamma,k,\setCount)$. All of the complexity results of
this paper at both model 0 and model 1 are independent of parameter
$\setSize$ .

\begin{definition}\label{app-list-def}
For a list $L$ of sets $A_1,A_2,\cdots, A_{\setCount}$ and real
$\alpha_L,\alpha_R,\delta_L,\delta_R\in [0,1)$, it is called
$((\alpha_L,\alpha_R),(\delta_L,\delta_R))$-list if each set $A_i$
is associated with a number $s_i$ with $(1-\delta_L)|A_i|\le s_i\le
(1+\delta_R)|A_i|$ for $i=1,2,\cdots, \setCount$, and the set $A_i$
has a $(\alpha_L,\alpha_R)$-biased random generator
\randomElm($A_i$).
\end{definition}

\section{Outline of Our Methods}\label{overview-sec}

For two sets $A$ and $B$, we develop a randomized method to
approximate the cardinality of the difference $B-A$.   We
approximate the size of $B-A$ by sampling a small number of elements
from $B$ and calculating the ratio of the elements in $B-A$ by
querying the set $A$. The approximate $|A\cup B|$ is the sum of an
approximation of  $|A|$ and an approximation of $|B-A|$.

A greedy approach will be based on the approximate difference
between a new set and the union of sets already selected. Assume
that $A_1,A_2,\cdots, A_{\setCount}$ is the list of sets for the
maximum coverage problem. After $A_{i_1},\cdots, A_{i_t}$ have been
selected, the greedy approach needs to check the size
$|A_j-(A_{i_1}\cup A_{i_2}\cup \cdots \cup A_{i_t})|$ before
selecting the next set. Our method to estimate $|A_j-(A_{i_1}\cup
A_{i_2}\cup \cdots \cup A_{i_t})|$ is based on randomization in
order to make the time independent of the sizes of input sets. Some
random samples are selected from set $A_j$.

The classical greedy approximation algorithm provides $1-(1-{1\over
k})^k$ ratio for the maximum coverage problem. The randomized greedy
approach gives $1-(1-{1\over k})^k-\xi$ ratio, where $\xi$ depends
on the accuracy of estimation to $|A_j-(A_{i_1}\cup A_{i_2}\cup
\cdots \cup A_{i_t})|$. As  $(1-{1\over k})^k$ is increasing and
${1\over e}=(1-{1\over k})^k+\Omega({1\over k})$, we can let
$(1-{1\over k})^k+\xi\le {1\over e}$ by using sufficient number of
random samples for the estimation of set difference when selecting a
new set. Thus, we control the accuracy of the approximate
cardinality of the set difference so that it is enough to achieve
the approximation ratio $1-{1\over e}$ for the maximum coverage
problem.

During the accuracy analysis, Hoeffiding Inequality
\cite{Hoeffding63} plays an important role. It shows how the number
of samples determines the accuracy of approximation.



\begin{theorem}[\cite{Hoeffding63}]\label{hoeffiding-theorem}
Let $X_1,\ldots , X_{\sampleCount}$ be $\sampleCount$ independent
random $0$-$1$ variables and $X=\sum_{i=1}^{\sampleCount} X_i$.

 i. If $X_i$ takes $1$ with probability at most $p$ for $i=1,\ldots ,
\sampleCount$, then for any $\epsilon>0$,
$\Pr(X>p\sampleCount+\epsilon \sampleCount)<e^{-{1\over
2}\sampleCount\epsilon^2}$.

ii. If  $X_i$ takes $1$ with probability at least $p$ for
$i=1,\ldots , \sampleCount$, then for any $\epsilon>0$,
$\Pr(X<p\sampleCount-\epsilon \sampleCount)<e^{-{1\over
2}\sampleCount\epsilon^2}$.
\end{theorem}

We define the function $\mu(x)$ in order to simply the probability
mentioned in
Theorem~\ref{hoeffiding-theorem}
\begin{eqnarray}
\mu(x)=e^{-{1\over 2}x^2}\label{function-mu(x)-def}
\end{eqnarray}

 Chernoff Bound
(see~\cite{MotwaniRaghavan00}) is also used in the maximum coverage
approximation when our main result is applied in some concrete
model. It implies a similar result as
Theorem~\ref{hoeffiding-theorem} (for example,
see~\cite{LiMaWang99}).

\begin{theorem}\label{chernoff3-theorem}
Let $X_1,\ldots , X_{\sampleCount}$ be $\sampleCount$ independent
random $0$-$1$ variables, where $X_i$ takes $1$ with probability at
least $p$ for $i=1,\ldots , \sampleCount$. Let
$X=\sum_{i=1}^{\sampleCount} X_i$, and $\mu=E[X]$. Then for any
$\delta>0$,
 $\Pr(X<(1-\delta)p\sampleCount)<e^{-{1\over 2}\delta^2 p\sampleCount}$.
\end{theorem}

\begin{theorem}\label{ourchernoff2-theorem}
Let $X_1,\ldots , X_{\sampleCount}$ be $\sampleCount$ independent
random $0$-$1$ variables, where $X_i$ takes $1$ with probability at
most $p$ for $i=1,\ldots , \sampleCount$. Let
$X=\sum_{i=1}^{\sampleCount} X_i$. Then for any $\delta>0$,
$\Pr(X>(1+\delta)p\sampleCount)<\left[{e^{\delta}\over
(1+\delta)^{(1+\delta)}}\right]^{p\sampleCount}$.
\end{theorem}





A well known fact in probability theory is the inequality
\begin{eqnarray}\Pr(E_1\cup E_2 \ldots \cup E_{\eventCount})\le
\Pr(E_1)+\Pr(E_2)+\ldots+\Pr(E_{\eventCount}),\label{prob-Add-Ineqn}
\end{eqnarray}
 where $E_1,E_2,\ldots, E_{\eventCount}$ are
$\eventCount$ events that may not be independent. In the analysis of
our randomized algorithm, there are multiple events such that the
failure from any of them may fail the entire algorithm. We often
characterize the failure probability of each of those events, and
use the above inequality to show that the whole algorithm has a
small chance to fail, after showing that each of them has a small
chance to fail.

Our algorithm performance will depend on the initial accuracy of
approximation to each set size, and how biased the random sample
from each input set. This consideration is based on the applications
to high dimensional geometry problems which may be hard to count the
exact number of elements in a set, and  is also hard to provide
perfect uniform random source. We plan to release more applications
to high dimensional geometry problems that need approximate counting
and biased random sampling.

Overall, our method is an approximate randomized greedy approach for
the maximum coverage problem. The numbers of random samples is
controlled so that it has enough accuracy to derive the classical
approximation ratio $1-{1\over e}$. The main results are stated at
Theorem~\ref{appr2-cover-theorem} (type 1 model) and
Corollary~\ref{appr2-cover-corol} (type 0 model).


\begin{definition}
Let the maximum  coverage problem have integer parameter $k$, and a
list $L$ of sets $A_1,A_2,\cdots, A_{\setCount}$ as input. We always
assume $k\le \setCount$ . Let $C^*(L,k)=|A_{t_1}\cup
A_{t_2}\cup\cdots\cup A_{t_k}|$ be the maximum union size of a
solution $A_{t_1},\cdots, A_{t_k}$ for the maximum coverage.
\end{definition}

\begin{theorem}\label{appr2-cover-theorem}Let $\rho$ be a constant in
$(0,1)$. For  parameters $\epsilon,\gamma\in (0,1)$ and
$\alpha_L,\alpha_R,\delta_L,\delta_R\in [0,1-\rho]$, there is an
algorithm to give a $(1-{1\over e^{\beta}})$ approximation for the
maximum cover problem, such that given a
$((\alpha_L,\alpha_R),(\delta_L,\delta_R))$-list $L$ of finite sets
$A_1,\cdots, A_{\setCount}$ and an integer $k$,  with probability at
least $1-\gamma$, it returns an integer $z$ and a subset $H\subseteq
\{1,2,\cdots, \setCount\}$ that satisfy
\begin{enumerate}[1.]
\item
$|\cup_{j\in H}A_j|\ge (1-{1\over e^{\beta}})C^*(L,k)$ and $|H|=k$,
\item
$((1-\alpha_L)(1-\delta_L)-\epsilon)|\cup_{j\in H}A_j|\le z\le
((1+\alpha_R)(1+\delta_R)+\epsilon)|\cup_{j\in H}A_j|$, and
\item
Its complexity is  $(T(\epsilon,\gamma,k,\setCount),
R(\epsilon,\gamma ,k,\setCount  ), Q(\epsilon,\gamma ,k,\setCount))$
with
\begin{eqnarray*}
T(\epsilon,\gamma,k,\setCount)&=&\bigO({k^5\over
\epsilon^2}(k\log({3 \setCount\over  k})+\log
{1\over \gamma})\setCount),\\
R(\epsilon,\gamma,k,\setCount)&=&\bigO({k^4\over \epsilon^2}(k\log
({3 \setCount\over  k})+\log
{1\over \gamma})\setCount),\ \ {\rm and} \\
Q(\epsilon,\gamma,k,\setCount)&=&\bigO({k^5\over
\epsilon^2}(k\log({3 \setCount\over  k})+\log {1\over
\gamma})\setCount),
\end{eqnarray*}
where $\beta={(1-\alpha_L)(1-\delta_L)\over
(1+\alpha_R)(1+\delta_R)}$.
\end{enumerate}
\end{theorem}

Corollary~\ref{appr2-cover-corol} gives the importance case that we
have exact sizes for all input sets, and uniform random sampling for
each of them. Such an input is called $((0,0),(0,0))$-list according
to Definition~\ref{app-list-def}.

\begin{corollary}\label{appr2-cover-corol} For  parameters $\epsilon,$ and $\gamma$ in $(0,1)$, there is a randomized algorithm
to give a $(1-{1\over e})$ approximation for the maximum  cover
problem, such that given a $((0,0),(0,0))$-list $L$ of finite sets
$A_1,\cdots, A_{\setCount}$ and an integer $k$, with probability at
least $1-\gamma$, it returns an integer $z$ and a subset $H\subseteq
\{1,2,\cdots, \setCount\}$ that satisfy
\begin{enumerate}[1.]
\item
$|\cup_{j\in H}A_j|\ge (1-{1\over e})C^*(L,k)$ and $|H|=k$,
\item
$(1-\epsilon)|\cup_{j\in H}A_j|\le z\le (1+\epsilon)|\cup_{j\in
H}A_j|$, and
\item
Its complexity is  $(T(\epsilon,\gamma,k,\setCount),
R(\epsilon,\gamma,k,\setCount), Q(\epsilon,\gamma,k,\setCount))$
with
\begin{eqnarray*}
T(\epsilon,\gamma,k,\setCount)&=&\bigO({k^5\over \epsilon^2}(k\log
({3 \setCount\over  k})+\log
{1\over \gamma})\setCount),\\
R(\epsilon,\gamma,k,\setCount)&=&\bigO({k^4\over \epsilon^2}(k\log
({3 \setCount\over  k})+\log
{1\over \gamma})\setCount),\ \ {\rm and} \\
Q(\epsilon,\gamma,k,\setCount)&=&\bigO({k^5\over \epsilon^2}(k\log
({3 \setCount\over  k})+\log {1\over \gamma})\setCount).
\end{eqnarray*}
\end{enumerate}
\end{corollary}

\begin{proof}
Since $\alpha_L=\alpha_R=\delta_L=\delta_R=0$ implies $\beta=1$, it
follows from Theorem~\ref{appr2-cover-theorem}.
\end{proof}

\section{Randomized Algorithm for Maximum Coverage}\label{union-sec}

We give a randomized algorithm for approximating the maximum
coverage. It is based on an approximation to the cardinality of set
difference. The algorithms are described at type 1 model, and has
corollaries for type 0 model.

\subsection{Randomized Algorithm for Set Difference
Cardinality}\label{set-diff-sec}

In this section, we develop a method to approximate the cardinality
of $B-A$ based on random sampling. It will be used as a submodule to
approximate the maximum coverage.




\begin{definition}\label{test-function-def}
Let $R=x_1,x_2,\cdots, x_w$ be a list of elements from set $B$, and
let $L$ be a list of sets $A_1,A_2,\cdots, A_u$. Define $\test(L,
R)=|\{j: 1\le j\le w, {\rm and}\ x_j\not\in (A_1\cup
A_2\cup\cdots\cup A_u)\}|$.
\end{definition}

 The Algorithm ApproximateDifference(.) gives an approximation $s$ for the
size of $B-A$. It is very time consuming to approximate $|B-A|$ when
$|B-A|$ is much less than $|B|$. The algorithm
ApproximateDifference(.) returns an approximate value $s$ for
$|B-A|$ with a range in  $[(1-\delta)|B-A|-\epsilon |B|,
(1+\delta)|B-A|+\epsilon |B|]$, and will not lose much accuracy when
it is applied to approximate the maximum coverage by controlling the
two parameters $\delta$ and $\epsilon$.


\newcounter{RandomTest}
\setcounter{RandomTest}{0}

\addtocounter{problem-counter}{1}
\addtocounter{RandomTest}{\arabic{problem-counter}}


\vskip 10pt

{\bf Algorithm~{\arabic{problem-counter} }:  RandomTest($L,B, w$)}

Input: $L$ is a list of sets $A_1, A_2,\cdots, A_u$, $B$ is another
set with a random generator \randomElm$(B)$, and $w$ is an integer
to control the number of random samples from $B$.

\begin{enumerate}[1.]

\item
\qquad For $i=1$ to $w$
 let $x_i=\randomElm(B)$;

\item
\qquad For $i=1$ to $w$

\item

\qquad\qquad
 Let $y_i=0$ if $(x_i \in A_1\cup A_2\cup \cdots\cup A_u)$, and
$1$ otherwise;

\item
\qquad Return $t=y_1+\cdots+y_w$;

\end{enumerate}

{\bf End of Algorithm}

 \vskip 10pt
\addtocounter{problem-counter}{1}

 {\bf Algorithm~{\arabic{problem-counter} }: ApproximateDifference($L,B$,$s_2, 
 \epsilon,\gamma$)}

Input: $L$ is a list of sets $A_1, A_2,\cdots, A_u$, $B$ is another
set with a random generator \randomElm$(B)$,
integer $s_2$ is an approximation for $|B|$ with
$(1-\delta_L)|B|\le s_2\le (1+\delta_R)|B|$, and $\epsilon$ and
$\gamma$ are real parameters in $(0,1)$, where $\delta\in [0,1]$.

{\bf Steps:}

\begin{enumerate}[1.]

\item\label{line-w-set-line}
\qquad Let $w$ be an integer with  $\mu({\epsilon\over 3})^w\le
{\gamma\over 4}$, where $\mu(x)$ is defined in equation
(\ref{function-mu(x)-def}).


\item
\qquad Let $t=$RandomTest($L,B, w$);

\item\label{alg-approx-diff-return-line}
\qquad
Return $s={t\over w}\cdot s_2$

\end{enumerate}

{\bf End of Algorithm}

\vskip 10pt

Lemma~\ref{difference-lemma} shows how Algorithm
ApproximateDifference(.) returns an approximation $s$  for $|B-A|$
with a small failure probability $\gamma$, and its complexity
depends on the accuracy $\epsilon$ of approximation and probability
$\gamma$. Its accuracy is controlled for the application to the
approximation algorithms for maximum coverage problem.

\begin{lemma}\label{difference-lemma}Assume that real number $\epsilon \in[0,1]$,  $B$ is a set with
$(\alpha_L,\alpha_R)$-biased random generator \randomElm$(B)$ and an
approximation $s_2$ for $|B|$ with $(1-\delta_L)|B|\le s_2\le
(1+\delta_R)|B|$, and $L$ is a list of sets $A_1, A_2,\cdots, A_u$.
Then
\begin{enumerate}[1.]
\item
If $R=x_1,x_2,\cdots, x_w$ be a list of elements generated by
\randomElm$(B)$, and $\mu({\epsilon\over 3})^w\le {\gamma\over 4}$,
then with probability at most $\gamma$, the  value $s={t\over
w}\cdot s_2$ fails to
 satisfy inequality (\ref{set-diff-approx-condition})
\begin{eqnarray}
(1-\alpha_L)(1-\delta_L)|B-A|-\epsilon|B|\le s\le
(1+\alpha_R)(1+\delta_R)|B-A|+\epsilon|B|,\label{set-diff-approx-condition}
\end{eqnarray}
where $A=A_1\cup A_2\cup\cdots\cup A_u$ is the union of sets in the
input list $L$.
\item
 With probability at most $\gamma$,  the returned value $s$ by the
algorithm ApproximateDifference(.) fails to
 satisfy inequality (\ref{set-diff-approx-condition}), and
\item If the implementation of RandomTest(.) in Algorithm~{\arabic{RandomTest}} is used,
then the complexity of ApproximateDifference(.) is
  $(T_D(\epsilon,\gamma,u),
R_D(\epsilon,\gamma,u), Q_D(\epsilon,\gamma,u))$ with
  $T_D(\epsilon,\gamma,u)=\bigO({u\over \epsilon^2}\log {1\over \gamma})$,
  $R_D(\epsilon,\gamma,u)=\bigO({1\over \epsilon^2}\log {1\over
  \gamma})$,
and
   $Q_D(\epsilon,\gamma,u)=\bigO({u\over \epsilon^2}\log {1\over \gamma})$.
\end{enumerate}
\end{lemma}


\ifdefined\RegularVersion

\begin{proof}
Let $A=A_1\cup A_2\cdots\cup A_u$. The $w$ random elements from $B$
are via the $(\alpha_L,\alpha_R)$-biased random generator
\randomElm$(B)$. We get $t$ to be the number of the $w$ items in
$B-A$. Value $s={t\over w}\cdot s_2$ is an approximation for
$|B-A|$. Let $p={|B-A|\over |B|}$, $p_L=(1-\alpha_L)p$, and
$p_R=(1+\alpha_R)p$.
By Theorem~\ref{hoeffiding-theorem}, with probability at most
$P_1=\mu({\epsilon\over 3})^w$,
we have $t>p_Rw+{\epsilon\over 3}\cdot
w=(1+\alpha_R)pw+{\epsilon\over 3}\cdot w$.

If $t\le (1+\alpha_R)pw+{\epsilon\over 3}\cdot  w$, then the value
\begin{eqnarray*}
s&=&{t\over w}\cdot s_2\le {(1+\alpha_R)pw+{\epsilon\over 3}\cdot
w\over w}\cdot
s_2\le ((1+\alpha_R)p+{\epsilon\over 3})s_2\\
&\le& ((1+\alpha_R)p+{\epsilon\over 3} )(1+\delta_R)|B|\le
(1+\alpha_R)(1+\delta_R)|B-A|+{\epsilon\over 3}\cdot
(1+\delta_R)|B|\\
&\le&
(1+\alpha_R)(1+\delta_R)|B-A|+\epsilon|B|.
\end{eqnarray*}

By Theorem~\ref{hoeffiding-theorem}, with probability at most
$P_2=\mu({\epsilon\over 3})^w$,
we have $t<p_Lw-{\epsilon\over 3}\cdot w=(1-\alpha)pw-{\epsilon\over
3}\cdot w$.

If $t\ge (1-\alpha_L)pw-{\epsilon\over 3}\cdot w$, then the value
\begin{eqnarray*}
s&=&{t\over w}\cdot s_2\ge {(1-\alpha_L)pw-{\epsilon\over 3}\cdot
w\over w}\cdot
s_2\ge ((1-\alpha_L)p-{\epsilon\over 3})s_2\\
&\ge& ((1-\alpha_L)p-{\epsilon\over 3})(1-\delta_L)|B|\ge
(1-\alpha_L)(1-\delta_L)|B-A|-{\epsilon\over 3}\cdot |B|\\
&\ge& (1-\alpha_L)(1-\delta_L)|B-A|-\epsilon |B|.
\end{eqnarray*}


By  line~\ref{line-w-set-line} of ApproximateDifference(.), we need
$w=\bigO({1\over \epsilon^2}\log {1\over \gamma})$ random samples in
$B$ so that the total failure probability is at most $P_1+P_2\le
2\cdot {\gamma\over 4}<\gamma$ (by
inequality~(\ref{prob-Add-Ineqn})).
The number of queries to $A$ is $w$. Thus, the number of total
queries to $A_1,A_2,\cdots, A_u$ is $uw$.


Therefore, we have its complexity $(T_D(\epsilon,\gamma),
R_D(\epsilon,\gamma), Q_D(\epsilon,\gamma))$ with
\begin{eqnarray*}
  T_D(\epsilon,\gamma)&=&\bigO(uw)=\bigO({u\over \epsilon^2}\log {1\over \gamma}),\\
  R_D(\epsilon,\gamma)&=&w=\bigO({1\over \epsilon^2}\log {1\over
  \gamma}),   \ {\rm and}\\
  Q_D(\epsilon,\gamma)&=&\bigO(uw)=\bigO({u\over \epsilon^2}\log {1\over \gamma}).
\end{eqnarray*}

This completes the proof of Lemma~\ref{difference-lemma}.
\end{proof}

\fi

\subsection{A Randomized Algorithm for Set Union Cardinality}

We describe a randomized algorithm for estimating the cardinality
for set union. It will use the algorithm for set difference
developed in Section~\ref{set-diff-sec}. The following lemma gives
an approximation for the size of sets union. Its accuracy is enough
when it is applied in the approximation algorithms for maximum
coverage problem.

\begin{lemma}\label{main-lemma} Assume $\epsilon, \delta_L,\delta_R,,\delta_{2,L},\delta_{2,R},\alpha_L,\alpha_R\in [0,1]$,
$(1-\delta_L)\le (1-\alpha_L)(1-\delta_{2,L})$ and $(1+\delta_R)\ge
(1+\alpha_R)(1+\delta_{2,R})$. Assume that $L$ is a list of sets
$A_1, A_2,\cdots, A_u$,  and $X_2$ is set with an
$(\alpha_L,\alpha_R)$-biased random generator \randomElm$(X_2)$. Let
integers $s_1$ and $s_2$ satisfy $(1-\delta_L)|X_1|\le s_1\le
(1+\delta_R)|X_1|$, and $(1-\delta_{2,L})|X_2|\le s_2\le
(1+\delta_{2,R})|X_2|$, then
\begin{enumerate}
\item\label{case1-diff}
If $t$ satisfies
$(1-\alpha_L)(1-\delta_{2,L})|X_2-X_1|-\epsilon|X_2|\le t \le
(1+\alpha_R)(1+\delta_{2,R})|X_2-X_1|+\epsilon|X_2|$, then $s_1+t$
satisfies
\begin{eqnarray}(1-\delta_L-\epsilon)|X_1\cup X_2|\le
s_1+t\le (1+\delta_R+\epsilon)|X_1\cup X_2|.\label{union-ineqn}
\end{eqnarray}
\item\label{case2-diff}
If $t=$ApproximateDifference($L,X_2, s_2,\epsilon, \gamma$),
 with probability at most $\gamma$,
$s_1+t$ does  not have inequality (\ref{union-ineqn}),
\end{enumerate}
where $X_1=A_1\cup A_2\cup \cdots\cup A_u$.
\end{lemma}

\ifdefined\RegularVersion

\begin{proof}
Assume that $s_1$ and $s_2$ satisfy
\begin{eqnarray}
(1-\delta_L)|X_1|&\le& s_1\le (1+\delta_R)|X_1|, {\rm\ \  and }\\
(1-\delta_{2,L})|X_2|&\le& s_2\le (1+\delta_{2,R})|X_2|.
\end{eqnarray}

Since $(1+\delta_R)\ge (1+\alpha_R)(1+\delta_{2,R})$, we have
\begin{eqnarray*}
s_1+t&\le& (1+\delta_R)|X_1|+(1+\alpha_R)(1+\delta_{2,R})|X_2-X_1|+\epsilon|X_2|\\
 &\le&
(1+\delta_R)(|X_1|+|X_2-X_1|)+\epsilon|X_2|\\
&=& (1+\delta_R)|X_1\cup X_2|+\epsilon |X_2|\\
&\le& (1+\delta_R+\epsilon)|X_1\cup X_2|.
\end{eqnarray*}

Since $(1-\delta_L)\le (1-\alpha_L)(1-\delta_{2,R})$, we have
\begin{eqnarray*}
s_1+t&\ge& (1-\delta_L)|X_1|+(1-\alpha_L)(1-\delta_{2,L})|X_2-X_1|-\epsilon|X_2|\\
 &\ge&
(1-\delta_L)(|X_1|+|X_2-X_1|)-\epsilon|X_2|\\
&=& (1-\delta_L)|X_1\cup X_2|-\epsilon|X_2|\\
&\ge& (1-\delta_L-\epsilon)|X_1\cup X_2|.
\end{eqnarray*}

Case~\ref{case2-diff} follows from Case~\ref{case1-diff}, and
Lemma~\ref{difference-lemma}. By executing
$t=$ApproximateDifference($X_1,X_2, s_2$,$\epsilon$, $\gamma$), we
have $(1-\alpha_L)(1-\delta_{2,L})|X_2-X_1|-\epsilon|X_2|\le t \le
(1+\alpha_R)(1+\delta_{2,R})|X_2-X_1|+\epsilon|X_2|$. The
probability to fail inequality (\ref{union-ineqn}) is at most
$\gamma$ by Lemma~\ref{difference-lemma}.
\end{proof}

\fi

\subsection{Approximation to the Maximum  Coverage
Problem}\label{max-cover-sec}

In this section, we show that our randomized approach to the
cardinality of set union can be applied to the maximum  coverage
problem. Lemma~\ref{cover-lemma} gives the approximation performance
of greedy method for the maximum coverage problem. It is adapted to
a similar result\cite{HochbaumPathria98} with our approximation
accuracy to the size of set difference.


\begin{definition}
For a list $L$ of sets $A_{1},A_{2},\cdots, A_{\setCount}$,  define
its initial $h$ sets by $L(h)=A_{1}, A_{2}, \cdots, A_{h}$, and the
union of sets in $L$ by $U(L)=A_{1}\cup A_{2}\cup \cdots\cup
A_{\setCount}$.
\end{definition}

\begin{lemma}\label{cover-lemma}Let $L'$ be a sublist of sets $A_{t_1}, A_{t_2},\cdots, A_{t_k}$ selected from the list $L$ of sets $A_1,A_2,\cdots, A_{\setCount}$. If each
subset $A_{t_{j+1}} (j=0,2,\cdots, k-1)$ in $L'$ satisfies
$|A_{t_{j+1}}-U(L'(j))|\ge \theta\cdot {C^*(L,k)-|U(L'(j))|\over
k}-\delta C^*(L,k)$,
then $|U(L'(l))|\ge (1-(1-{\theta\over k})^l)C^*(L,k)-l\cdot \delta
C^*$ for $l=1,2,\cdots, k$.
\end{lemma}

\ifdefined\RegularVersion

\begin{proof}
It is proven by induction. It is trivial at $l=1$ as
$L'(0)=\emptyset$. Assume $|U(L'(l))|\ge (1-(1-{\theta\over
k})^l)C^*(L,k)-l\cdot \delta C^*(L,k)$. Consider the case $l+1$.

Let $A_{t_{l+1}}$ satisfy $|A_{t_{l+1}}-U(L'(l))|\ge\theta\cdot
{C^*(L,k)-|U(L'(l))|\over k}-\delta C^*(L,k)$.

Therefore,
\begin{eqnarray*}
|U(L'(l+1))|&=&|U(L'(l))|+|A_{t_{l+1}}-U(L'(l))|\\
&\ge&|U(L'(l))|+ \theta\cdot {C^*(L,k)-|U(L'(l))|\over k}-\delta
C^*(L,k)\\
&=&(1-{\theta\over k})|U(L'(l))|+ {\theta C^*(L,k)\over k} -\delta
C^*(L,k)\\
&\ge& (1-{\theta\over k})((1-(1-{\theta\over k})^l)C^*(L,k)-l\delta
C^*(L,k))+{\theta C^*(L,k)\over k}- \delta
C^*(L,k)\\
&\ge& (1-{\theta\over k})(1-(1-{\theta\over
k})^l)C^*(L,k)-(1-{\theta\over k})\cdot l\cdot \delta
C^*(L,k)+{\theta C^*(L,k)\over k}- \delta
C^*(L,k)\\
&\ge& (1-{\theta\over k})(1-(1-{\theta\over k})^l)C^*(L,k)+{\theta
C^*(L,k)\over k}-l\cdot \delta C^*(L,k)- \delta
C^*(L,k)\\
&=&(1-(1-{\theta\over k})^{l+1})C^*(L,k)-(l+1)\cdot \delta C^*(L,k).
\end{eqnarray*}
\end{proof}

\fi

\begin{definition}
If $L'$ is a list of sets $B_1,B_2,\cdots, B_u$, and $B_{u+1}$ is
another set, define Append$(L', B_{u+1})$ to be the list
$B_1,B_2,\cdots, B_u, B_{u+1}$, which is to append $B_{u+1}$ to the
end of $L'$.
\end{definition}

In Algorithm ApproximateMaximumCover(.), there are several virtual
functions including RandomSamples(.),
ApproximateSetDifferenceSize(.), and ProcessSet(.), which have
variant implementations and will be given in Virtual Function
Implementations 1,2 and 3. We use a virtual function
ApproximateSetDifferenceSize$(L', A_i,s_i, \epsilon', \gamma,
k,\setCount)$ to approximate $|A_i-\cup_{A_j\ is \ in\ L'}A_j|$. We
will have variant implementations for this function, and get
different time complexity. One implementation will be given at
Lemma~\ref{appr-cover-complexity-lemma}, and the other one will be
given at Lemma~\ref{improve-ramdom-sample-lemma}. Another function
ProcessSet($A_j)$ also has variant implementations. Its purpose is
to process a newly selected set $A_j$ to list $L'$ of existing
selected sets, and may sort it in one of the implementations. The
function RandomSamples$(.)$ is also virtual and will have two
different implementations.

\vskip 10pt


\addtocounter{problem-counter}{1}

\newcounter{problem-counter-max-cover}
\setcounter{problem-counter-max-cover}{0}


\addtocounter{problem-counter-max-cover}{\arabic{problem-counter}}

 {\bf Algorithm~{\arabic{problem-counter} }:} ApproximateMaximumCover($L,k,\xi,\gamma)$

Input: a list $((\alpha_L,\alpha_R),(\delta_L,\delta_R))$-list $L$
of $\setCount$ sets $A_1, A_2,\cdots, A_{\setCount}$,
an integer parameter $k$, and two real parameters $\xi,\gamma\in
(0,1)$. Each $A_i$ has a $(\alpha_L,\alpha_R)$-biased random
generator \randomElm$(A_i)$, and an approximation $s_i$ for $|A_i|$.

{\bf Steps:}

\begin{enumerate}[1.]

\item\label{initial-start}
Let $H=\emptyset$, and list $L'$ be empty;

\item
Let $z=0$;

\item\label{epsilon'-line}
Let $\epsilon'={\xi\over 4k}$;

\item
For $i=1$ to $\setCount$ let $R_i=$RandomSamples$(A_i,
\xi,\gamma,k,\setCount)$;


\item\label{For-loop-start}
For $j=1$ to $k$

\item $\{$

\item
\qquad Let $s_j^*=-1$;\label{sj-init-line}

\item\label{call-app-diff-line-max-cover}
\qquad For each $A_i$ in $L$,

\item
\qquad $\{$

\item\label{appr-diff-line}
 \qquad\qquad Let $s_{i,j}=$ApproximateSetDifferenceSize$(L',
A_i,s_i, R_i, \epsilon', \gamma, k,\setCount)$;


\item
\qquad\qquad If $(s_{i,j}>s_j^*)$ then let $s_j^*=s_{i,j}$ and
$t_j=i$;


\item
\qquad $\}$

\item
\qquad Let $H=H\cup \{t_j\}$;

\item
\qquad Let $z=z+s_{t_j,j}$;

\item
\qquad ProcessSet$(A_{t_j})$;

\item\label{merege-line}
 \qquad Let $L'=$Append($L', A_{t_j})$;

\item
\qquad Remove $A_{t_j}$ from list $L$;

\item\label{call-app-diff-line-max-cover2} $\}$


\item
Return $z$ and $H$;

\end{enumerate}

{\bf End of Algorithm}

\vskip 10pt

The algorithm ApproximateMaximumCover(.) is a randomized greedy
approach for the maximum coverage problem. It adds the set $A_{t_j}$
that has an approximate largest $|A_{t_j}-(\cup_{A_i\in L'} A_i)|$
to the existing partial solution saved in the list $L'$. The
accuracy control for the estimated size of the set
$A_{t_j}-(\cup_{A_i\in L'} A_i)$ will be enough to achieve the same
approximation ratio as the classical deterministic algorithm. Since
$s_j^*$ starts from $-1$ at line~\ref{sj-init-line} in the algorithm
ApproximateMaximumCover(.), each iteration picks up one set from the
input list $L$, and remove it from $L$. By the end of the algorithm,
$L'$ contains $k$ sets if $k\le \setCount $.

Lemma~\ref{appr1-cover-lemma} shows the approximation accuracy for
the maximum coverage problem if $s_{i,j}$ is accurate enough to
approximate $|A_i- U(L')|$. It may be returned by
ApproximateDifference(.) with a small failure probability and
complexity shown at Lemma~\ref{appr-cover-complexity-lemma}.

\begin{lemma}\label{appr1-cover-lemma}
Let $\xi\in (0,1)$, and $\alpha_L,\alpha_R,\delta_L,\delta_R\in
[0,1)$,  $L$ be a $((\alpha_L,\alpha_R),(\delta_L,\delta_R))$-list
of sets $A_1,\cdots, A_{\setCount}$, and $L'_*$ be the sublist $L'$
after algorithm ApproximateMaximumCover(.) is completed.  If every
time $s_{i,j}$ in the line~\ref{appr-diff-line} of the algorithm
ApproximateMaximumCover(.) satisfies $$
(1-\alpha_L)(1-\delta_L)|A_i-U(L'_*(j))|-\epsilon'|A_i|\le
s_{i,j}\le
(1+\alpha_R)(1+\delta_R)|A_i-U(L'_*(j))|+\epsilon'|A_i|,$$
 then it
returns an integer $z$ and a size $k$ subset $H\subseteq
\{1,2,\cdots, \setCount\}$ that satisfy
\begin{enumerate}[1.]
\item
$|\cup_{j\in H}A_j|\ge (1-(1-{\beta\over k})^k-\xi )C^*(L,k)$, and
\item\label{case-union-approx}
$((1-\alpha_L)(1-\delta_L)-\xi)|\cup_{j\in H}A_j|\le z\le
((1+\alpha_R)(1+\delta_R)+\xi)|\cup_{j\in H}A_j|$,
\end{enumerate}
where $\beta={(1-\alpha_L)(1-\delta_L)\over
(1+\alpha_R)(1+\delta_R)}$.
\end{lemma}

\ifdefined\RegularVersion

\begin{proof}
Each time the randomized greedy algorithm selects a set from the
input list that is close to have the best improvement for coverage.
Let $L'_*$ be the list $L'$ after the algorithm
ApproximateMaximumCover(.) is completed ($L'$ is dynamic list during
the execution of the algorithm, and $L'_*$ is static after the
algorithm ends). The list $L'_*$ contains $k$ subsets $A_{t_1},
\cdots, A_{t_k}$. According to the algorithm, $L'_*(j)$ is the list
of $j$ subsets $A_{t_1}, \cdots, A_{t_j}$ that have been appended to
$L'$ after the first $j$ iterations for the loop from
line~\ref{For-loop-start} to
line~\ref{call-app-diff-line-max-cover2}.

Assume that for each $s_{i,j}$
we have
\begin{eqnarray}
(1-\alpha_L)(1-\delta_L)|A_i-U(L'_*(j))|-\epsilon'|A_i|\le
s_{i,j}\le (1+\alpha_R)(1+\delta_R)|A_i-U(L'_*(j))|+\epsilon'|A_i|
\label{Aj-sj-ineqn}
\end{eqnarray}

If set $A_{v_j}$ makes $|A_{v_j}-U(L'_*(j))|$ be the largest, we
have
\begin{eqnarray}
|A_{v_j}-U(L'_*(j))|\ge {C^*(L,k)-|U(L'_*(j))|\over
k}.\label{Aj0-C*-ineqn}
\end{eqnarray}

A special case of inequality~(\ref{Aj-sj-ineqn}) is inequality
(\ref{Aj0-sj0-ineqn})
\begin{eqnarray}
(1-\alpha_L)(1-\delta_L)|A_{v_j}-U(L'_*(j))|-\epsilon'|A_{v_j}|\le
s_{v_j,j}\le
(1+\alpha_R)(1+\delta_R)|A_{v_j}-U(L'_*(j))|+\epsilon'|A_{v_j}|.\label{Aj0-sj0-ineqn}
\end{eqnarray}

Since $s_{v_j,j}\le s_{t_j,j}$, we have
$(1-\alpha_L)(1-\delta_L)|A_{v_j}-U(L'_*(j))|-\epsilon'|A_{v_j}|\le
s_{v_j,j}\le s_{t_j,j}\le
(1+\alpha_R)(1+\delta_R)|A_{t_j}-U(L'_*(j))|+\epsilon'|A_{t_j}|$ by
inequalities (\ref{Aj-sj-ineqn}) and (\ref{Aj0-sj0-ineqn}).

Therefore,
$(1-\alpha_L)(1-\delta_L)|A_{v_j}-U(L'_*(j))|-\epsilon'|A_{v_j}|\le
(1+\alpha_R)(1+\delta_R)|A_{t_j}-U(L'_*(j))|+\epsilon'|A_{t_j}|$.
Since $|A_{v_j}|\le C^*(L,k)$ and $|A_{t_j}|\le C^*(L,k)$, we have
${(1-\alpha_L)(1-\delta_L)\over (1+\alpha_R)(1+\delta_R)}\cdot
|A_{v_j}-U(L'_*(j))|-{2\epsilon'\over
(1+\alpha_R)(1+\delta_R)}C^*(L,k)\le |A_{t_j}-U(L'_*(j))|$. By
inequality (\ref{Aj0-C*-ineqn}), we have
\begin{eqnarray*}
\beta\cdot {C^*(L,k)-|U(L'_*(j))|\over k}-2\epsilon'C^*(L,k)\le
|A_{t_j}-U(L'_*(j))|,
\end{eqnarray*}
 where $\beta={(1-\alpha_L)(1-\delta_L)\over
(1+\alpha_R)(1+\delta_R)}$.

By Lemma~\ref{cover-lemma} and line~\ref{epsilon'-line} in
ApproximateMaximumCover(.), we have $L'_*$ with
\begin{eqnarray*}
|U(L'_*)|&\ge& (1-(1-{\beta\over k})^k)C^*(L,k)-k\cdot 2\epsilon'C^*(L,k)\\
&\ge& (1-(1-{\beta\over k})^k-\xi ) C^*(L,k).
\end{eqnarray*}

Case~\ref{case-union-approx} of this lemma can be proven by a simple
induction. We just need to prove after the $i$-th iteration of the
for loop from line~\ref{For-loop-start} to
line~\ref{call-app-diff-line-max-cover2} of this algorithm,
\begin{eqnarray}
((1-\alpha_L)(1-\delta_L)-i\epsilon')|\cup_{j\in H}A_j|\le z\le
((1+\alpha_R)(1+\delta_R)+i\epsilon')|\cup_{j\in
H}A_j|.\label{case-2-inequality}
\end{eqnarray}

It is trivial right after the initialization
(line~\ref{initial-start} to line~\ref{epsilon'-line}
of the algorithm) since $H=\emptyset$ and $z=0$. Assume inequality
(\ref{case-2-inequality}) holds after the $i$-th iteration of the
loop from line~\ref{For-loop-start} to
line~\ref{call-app-diff-line-max-cover2}.
By inequality (\ref{Aj-sj-ineqn}) and Lemma~\ref{main-lemma}
we have following inequality (\ref{induction-ineqn}) after the
$(i+1)$-th iteration of the loop from line~\ref{For-loop-start} to
line~\ref{call-app-diff-line-max-cover2}.

\begin{eqnarray}
((1-\alpha_L)(1-\delta_L)-(i+1)\epsilon')|\cup_{j\in H}A_j|\le z\le
((1+\alpha_R)(1+\delta_R)+(i+1)\epsilon')|\cup_{j\in
H}A_j|.\label{induction-ineqn}
\end{eqnarray}

Thus, at the end of the algorithm, we have
\begin{eqnarray}
((1-\alpha_L)(1-\delta_L)-k\epsilon')|\cup_{j\in H}A_j|\le z\le
((1+\alpha_L)(1+\delta_R)+k\epsilon')|\cup_{j\in H}A_j|.
\end{eqnarray}

Thus, by the end of the algorithm,
we have the following inequality (\ref{induction-final-ineqn}):

\begin{eqnarray}
((1-\alpha_L)(1-\delta_L)-\xi)|\cup_{j\in H}A_j|\le z\le
((1+\alpha_R)(1+\delta_R)+\xi)|\cup_{j\in
H}A_j|.\label{induction-final-ineqn}
\end{eqnarray}

\end{proof}

\fi

We need Lemma~\ref{help-app-lemma} to transform the approximation
ratio given by Lemma~\ref{appr1-cover-lemma} to constant $(1-{1\over
e})$ to match the classical ratio for the maximum coverage problem.

\begin{lemma}\label{help-app-lemma}
For each integer $k\ge 2$, and real $b\in [0,1]$, we have
\begin{enumerate}[1.]
\item\label{first-help-app-lemma}
$(1-{b\over k})^k\le {1\over e^{b}}-{\eta b\over 2e^{b}k}$, and
\item\label{secon-help-app-lemma}
If $\xi\le {\eta b \over 4e^{b}k}$, then
 $1-(1-{b\over k})^k-\xi> 1-{1\over
e^{b}}$,  where $\eta=e^{-{1\over 4}}$.
\end{enumerate}
\end{lemma}

\ifdefined\RegularVersion

\begin{proof}
Let function $f(x)=1-\eta x-e^{-x}$. We have $f(0)=0$. Taking
differentiation, we get ${d f(x)\over dx}=-\eta +e^{-x}>0$ for all
$x\in (0,{1\over 4})$.

Therefore, for all $x\in (0,{1\over 4})$,
\begin{eqnarray}
e^{-x}\le 1-\eta x.\label{exp-expansion-ineqn}
\end{eqnarray}
The following Taylor expansion can be found in standard calculus
textbooks. For all $x\in (0,1)$,
\begin{eqnarray}
\ln (1-x)=-x-{x^2\over 2}-{x^3\over 3}-\cdots .
\end{eqnarray}

Therefore, we have
\begin{eqnarray}
(1-{b\over k})^k&=&e^{k\ln (1-{b\over k})}=e^{k(-{b\over
k}-{b^2\over 2k^2}-{b^3\over 3k^3}-\cdots)}
=e^{-b-{b^2\over 2k}-{b^3\over 3k^2}-\cdots}\\
&\le&e^{-b-{b\over 2k}}=e^{-b}\cdot e^{-{b\over 2k}}\label{maximal-cov-ineqn1}\\
&\le& e^{-b}\cdot (1-\eta \cdot {b\over 2k})\le{1\over e^{b}}-{\eta
b \over 2e^{b}k}\label{maximal-cov-ineqn2}.
\end{eqnarray}
Note that the transition from~(\ref{maximal-cov-ineqn1})
to~(\ref{maximal-cov-ineqn2}) is based on
inequality~(\ref{exp-expansion-ineqn}).

The part~\ref{secon-help-app-lemma} follows from
part~\ref{first-help-app-lemma}. This is because $1-(1-{b\over
k})^k-\xi\ge 1-{1\over e^{b}}+{\eta b \over 2e^{b}k}-\xi\ge
1-{1\over e^{b}}+{\eta b \over 4e^{b}k}$.
\end{proof}

\fi

\subsection{Multiple Rounds Random Sampling for Maximum
Coverage}\label{multi-round-sec}

Theorem~\ref{appr-cover-theorem} gives the performance of our
randomized greedy approximation algorithm for the maximum coverage
problem. It uses multiple rounds of random samplings since there is
a series of executions for calling ApproximateDifference(.) via
ApproximateSetDifferenceSize$(.)$, which is given at Virtual
Function Implementation 1. This shows maximum coverage problem has a
$\setCount\poly(k)$ time $(1-{1\over e})$-approximation for
$((0,0),(0,0))$-list as input (see Definition~\ref{app-list-def})
under the model that each input set $A_i$ provides $\bigO(1)$-time
to generate a random sample and answer one membership query.

\vskip 10pt

\addtocounter{problem-counter}{1}

 {\bf Algorithm~{\arabic{problem-counter} }: Virtual Function Implementation 1}

The parameters $L', A_i,s_i, R_i, \epsilon', \gamma, k,\setCount$
follow from those in Algorithm~{\arabic{problem-counter-max-cover}}.

\qquad RandomSamples$(A_i, R_i, \xi,\gamma,k,\setCount)$;

\qquad \{

\qquad \qquad Let $R_i=\emptyset$;

\qquad \}

\vskip 10pt

\qquad RandomTest(.)
\{ the same as that defined at Algorithm~{\arabic{RandomTest}}  \}

\vskip 10pt

\qquad ApproximateSetDifferenceSize$(L', A_i,s_i, R_i, \epsilon',
\gamma, k,\setCount)$;

\qquad \{

\qquad \qquad Let $\gamma'={\gamma\over 4k\setCount}$;

\qquad \qquad Return ApproximateDifference$(L', A_i, s_i,
R_i,\epsilon', \gamma')$;

\qquad \}

\vskip 10pt

\qquad  ProcessSet($A_i)$ \{   \} (Do nothing to set $A_i$)

\vskip 10pt

{\bf End of Algorithm}

\vskip 10pt

Lemma~\ref{appr-cover-complexity-lemma} gives the complexity of the
ApproximateMaximumCover(.) using multiple rounds of random samplings
from the input sets. It also gives a small failure probability of
the algorithm. Its complexity will be improved and shown at
Lemma~\ref{improve-ramdom-sample-complexity-lemma} in
Section~\ref{improved-sect}.

\begin{lemma}\label{appr-cover-complexity-lemma}
Let $\xi\in (0,1)$, and $\alpha_L,\alpha_R,\delta_L,\delta_R\in
[0,1)$. Assume that the algorithm ApproximateMaximumCover(.) is
executed with Virtual Function Implementation 1. Let $L'_*$ be the
list $L'$ after the completion of ApproximateMaximumCover(.)
Then
\begin{enumerate}[1.]
\item
With probability at most $\gamma$, there is a value $s_{i,j}$ in the
line~\ref{appr-diff-line} of the algorithm
ApproximateMaximumCover(.) does not satisfies
$$(1-\alpha_L)(1-\delta_L)|A_i-U(L'_*(j))|-\epsilon'|A_i|\le
s_{i,j}\le (1+\alpha_R)(1+\delta_R)|A_i-U(L'_*(j))|+\epsilon'|A_i|,\
\ {\rm and}$$
\item
The algorithm has  complexity $(T_1(\xi,\gamma,k,\setCount),
R_1(\xi,\gamma,k,\setCount), Q_1(\xi,\gamma,k,\setCount)$ with
\begin{eqnarray*}
T_1(\xi,\gamma,k,\setCount)&=&k\setCount T_D({\xi\over
4k},{\gamma\over 4k\setCount},k),\\
  R_1(\xi,\gamma,k,\setCount)&=&k\setCount R_D({\xi\over 4k},{\gamma\over 4k\setCount},k), \ \ {\rm
  and} \\
  Q_1(\xi,\gamma,k,\setCount)&=&k\setCount Q_D({\xi\over 4k},{\gamma\over 4k\setCount},k),
\end{eqnarray*}
  where $(T_D(.), R_D(.), Q_D(.))$ are the complexity
  functions defined in Lemma~\ref{difference-lemma}, $\beta$ is the same as that in Lemma~\ref{appr1-cover-lemma},
   and $\epsilon'$ is the same as that in ApproximateMaximumCover(.).
\end{enumerate}
\end{lemma}

\ifdefined\RegularVersion

\begin{proof}
By Lemma~\ref{difference-lemma},  for each $s_{i,j}$
we have
\begin{eqnarray}
(1-\alpha_L)(1-\delta_L)|A_i-U(L'_*(j))|-\epsilon'|A_i|\le
s_{i,j}\le (1+\alpha_R)(1+\delta_R)|A_i-U(L'_*(j))|+\epsilon'|A_i|
\end{eqnarray}
 with probability at most $\gamma'$ (defined in Virtual Function Implementation 1 ) to fail. The total  probability
 that one of the $k\setCount$ cases fails
 is
at most $P_1=k\setCount\gamma'=k\setCount\cdot {\gamma\over
4k\setCount}\le {\gamma\over 4}$.

By Lemma~\ref{difference-lemma}, the function
ApproximateDifference($.$) at
line~\ref{call-app-diff-line-max-cover} has complexity
\begin{eqnarray*}
(T_D(\epsilon',\gamma',k), R_D(\epsilon',\gamma',k),
Q_D(\epsilon',\gamma',k)).
\end{eqnarray*}


The algorithm's complexity is
$(T_1(\xi,\epsilon,\gamma,k,\setCount),
R_1(\xi,\epsilon,\gamma,k,\setCount),
Q_1(\xi,\epsilon,\gamma,k,\setCount))$ with
\begin{eqnarray*}
 T_1(\xi,\epsilon,\gamma,k,\setCount)&=&k\setCount T_D(\epsilon',\gamma',k)=k\setCount T_D({\xi\over
4k},{\gamma\over 4k\setCount},k),\\
  R_1(\xi,\epsilon,\gamma,k,\setCount)&=&k\setCount R_D(\epsilon',\gamma',k)=k\setCount R_D({\xi\over 4k},{\gamma\over 4k\setCount},k), \ \ {\rm and}\\
  Q_1(\xi,\epsilon, \gamma,k,\setCount)&=&k\setCount Q_D(\epsilon',\gamma',k)=k \setCount Q_D({\xi\over 4k},{\gamma\over 4k\setCount},k).
\end{eqnarray*}
\end{proof}

\fi

\begin{theorem}\label{appr-cover-theorem}
Let $\rho$ be a constant in $(0,1)$. For  parameters
$\epsilon,\gamma\in (0,1)$ and
$\alpha_L,\alpha_R,\delta_L,\delta_R\in [0,1-\rho]$, there is an
algorithm to give a $(1-{1\over e^{\beta}})$-approximation for the
maximum cover problem ($\beta$ is defined in
Lemma~\ref{appr1-cover-lemma}), such that given a
$((\alpha_L,\alpha_R),(\delta_L,\delta_R))$-list $L$ of finite sets
$A_1,\cdots, A_{\setCount}$  and an integer $k$,  with probability
at least $1-\gamma$, it returns an integer $z$ and a subset
$H\subseteq \{1,2,\cdots, \setCount\}$ that satisfy
\begin{enumerate}[1.]
\item
$|\cup_{j\in H}A_j|\ge (1-{1\over e^{\beta}})C^*(L,k)$  and $|H|=k$,
\item
$((1-\alpha_L)(1-\delta_L)-\epsilon)|\cup_{j\in H}A_j|\le z\le
((1+\alpha_R)(1+\delta_R)+\epsilon)|\cup_{j\in H}A_j|$, and
\item
Its complexity is  $(T_C(\epsilon,\gamma,k,\setCount),
R_C(\epsilon,\gamma,k,\setCount), Q_C(\epsilon,\gamma,k,\setCount))$
where
\begin{eqnarray*}
T_C(\epsilon,\gamma,k,\setCount)&=&\bigO({k^6\setCount\over \epsilon^2}(\log \setCount+\log{1\over \gamma})),\\
  R_C(\epsilon,\gamma,k,\setCount)&=&\bigO({k^5\setCount\over \epsilon^2}(\log \setCount+\log{1\over
\gamma})), \ \ {\rm
  and} \\
  Q_C(\epsilon,\gamma,k,\setCount)&=&\bigO({k^6\setCount\over \epsilon^2}(\log \setCount+\log{1\over
\gamma})).
\end{eqnarray*}
\end{enumerate}
\end{theorem}

\ifdefined\RegularVersion

\begin{proof}
Select $\xi=\min({\epsilon\eta\beta \over 4e^{\beta}k},\epsilon)$,
where $\eta$ is defined in Lemma~\ref{help-app-lemma}. The accuracy
of approximation follows from Lemma~\ref{help-app-lemma},
Lemma~\ref{appr1-cover-lemma} and
Lemma~\ref{appr-cover-complexity-lemma}. The complexity follows from
the complexity functions $T_D(.), R_D(.),$ and $Q_D(.)$ in
Lemma~\ref{difference-lemma}. Since
$T_D(\epsilon,\gamma,k)=\bigO({k\over \epsilon^2}\log {1\over
\gamma})$, we have
\begin{eqnarray}
T_D({\xi \over 4k},{\gamma\over
4k\setCount},k)&=&\bigO({k^3\over\xi^2}\log
{4k\setCount\over \gamma})\\
&=&\bigO({k^5e^{2\beta}\over \epsilon^2\beta^2}(\log
\setCount+\log{1\over
\gamma}))\\
&=&\bigO({k^5\over \epsilon^2}(\log \setCount+\log{1\over \gamma})).
\end{eqnarray}

Similarly,
\begin{eqnarray}
R_D({\xi \over 4k},{\gamma\over 4k\setCount},k)&=&\bigO({k^4\over
\epsilon^2}(\log \setCount+\log{1\over \gamma}))\ \ \ {\rm and}\\
Q_D({\xi \over 4k},{\gamma\over 4k\setCount},k)&=&\bigO({k^5\over
\epsilon^2}(\log \setCount+\log{1\over \gamma})).
\end{eqnarray}

Thus,
\begin{eqnarray*}
T_C(\epsilon,\gamma,k,\setCount)&=&k\setCount T_D({\xi \over
4k},{\gamma\over
4k\setCount},k)=\bigO({k^6\setCount\over \epsilon^2}(\log \setCount+\log{1\over \gamma})),\\
  R_C(\epsilon,\gamma,k,\setCount)&=&k\setCount R_D({\xi \over 4k},{\gamma\over 4k\setCount},k)=\bigO({k^5\setCount\over
\epsilon^2}(\log \setCount+\log{1\over \gamma})), \ \ {\rm
  and} \\
  Q_C(\epsilon,\gamma,k,\setCount)&=&k\setCount Q_D({\xi \over 4k},{\gamma\over 4k\setCount},k)=\bigO({k^6\setCount\over
\epsilon^2}(\log \setCount+\log{1\over \gamma})).
\end{eqnarray*}
\end{proof}

\fi

\ifdefined\RegularVersion

\else
\let\clearpage\relax
\include{NonregularEnd}

\section{Technical Proofs for Section~\ref{union-sec}}

In this section, we give the proofs for the Lemmas and Theorems in
Section~\ref{union-sec}.

\subsection{Proof for Lemma~\ref{difference-lemma}}

\include{ProofForDifferenceLemma}

\subsection{Proof for Lemma~\ref{main-lemma}}
\include{ProofForMainLemma}

\subsection{Proof for Lemma~\ref{cover-lemma}}
\include{ProofForCoverLemma}

\subsection{Proof for Lemma~\ref{appr1-cover-lemma}}
\include{ProofForApp1CoverLemma}

\subsection{Proof for Lemma~\ref{help-app-lemma}}
\include{ProofForHelpAppLemma}

\subsection{Proof for Lemma~\ref{appr-cover-complexity-lemma}}
\include{ProofForApprCoverComplexityLemma}

\subsection{Proof for Theorem~\ref{appr-cover-theorem}}
\include{ProofForApprCoverTheorem}

\fi

\section{Faster Algorithm for Maximum
Coverage}\label{improved-sect}

In this section, we describe an improved approximation algorithm for
the maximum coverage problem. It has slightly less time and keeps
the same approximation ratio. We showed the multi-round random
sampling approach in Section~\ref{multi-round-sec} with
Theorem~\ref{appr-cover-theorem}. A single round random sampling
approach is given in this section with an improved time complexity.
It may help us under
how two different approaches affect the algorithm complexity.


In this section, we control the total number of random samples from
each set. The random samples from each set $S_i$ will be generated
in the beginning of algorithm. We show how to remove the samples
that are already in the selected sets saved in the list $L'$ of the
algorithm for ApproximateMaximumCover(.).

\newcommand{\sufficient}{{\rm sufficient random bound}}

\begin{definition}\label{shared-sample-def}
Assume that $\epsilon, \gamma\in (0,1)$ and
$k,\setCount\in\integerSet $. Let $L$ be a list of sets
$A_1,A_2,\cdots, A_{\setCount}$.
\begin{itemize}
\item
Define $h^*(k,\setCount)$ to be the number of subsets of size at
most $k$ in $\{1,2,\cdots, \setCount\}$.
\item
Define $\gamma_{k,\setCount}={\gamma\over nh^*(k,\setCount)}$.
\item
Define $g(\epsilon,
\gamma,k,\setCount)=R_D(\epsilon,\gamma_{k,\setCount },k)$, where
$R_D(.)$ is defined in Lemma~\ref{difference-lemma}.

\end{itemize}
\end{definition}



\begin{lemma}\label{improve-ramdom-sample-lemma}Assume parameters $\epsilon, \gamma,\alpha_L,\alpha_R,\delta_L,\delta_R\in (0,1)$ and
$k,\setCount\in\integerSet $. Let function $g(\epsilon,
\gamma,k,\setCount)$ be defined as in
Definition~\ref{shared-sample-def}. Let $L$ be a list of sets
$A_1,A_2,\cdots, A_{\setCount}$ such that each $A_i$ has a
$(\alpha_L,\alpha_R)$-biased random generator $\randomElm(A_i)$, and
an approximation $s_j$ for $|A_i|$ with $(1-\delta_L)|A_i|\le s_j\le
(1+\delta_R)|A_i|$.
\begin{enumerate}[1.]
\item\label{s2-lem} The function $g(.)$ has
$g(\epsilon,\gamma, k,\setCount)=\bigO({1\over \epsilon^2}(k\log
({3\setCount\over k})+\log {1\over \gamma}))$.
\item\label{s1-lem}
Let $R_i$ be a list of $w=g(\epsilon, \gamma,k,\setCount)$ random
samples from each set $A_i$ in the input list via the
$(\alpha_L,\alpha_R)$-biased generator $\randomElm(A_i)$, then with
failure probability at most $\gamma$, the value $s={r_{i,j}\over
w}\cdot s_i$ with $r_{i,j}=\test(L^*,R_i)$ (see
Definition~\ref{test-function-def})
satisfies the inequality~(\ref{set2-diff-approx-condition})
\begin{eqnarray}
(1-\alpha_L)(1-\delta_L)|A_i-A|-\epsilon|A_i|\le s\le
(1+\alpha_R)(1+\delta_R)|A_i-A|+\epsilon|A_i|,\label{set2-diff-approx-condition}
\end{eqnarray}
for every sublist $L^*=A_{t_1}, A_{t_2},\cdots, A_{t_j}$ with $j\le
k$ of $L$, where $A=A_{t_1}\cup A_{t_2}\cup\cdots\cup A_{t_j}$.

\end{enumerate}
\end{lemma}

\begin{proof}Let $A_1,\cdots, A_{\setCount}$ be the input list of sets, and $k$ be the integer parameter in the input.
 Let $V_k$ be the class of subsets from $\{1,2,\cdots, \setCount\}$ with
size at most $k$. In other words, $V_k=\{H: H\subseteq \{1,2,\cdots,
\setCount\}\ {\rm and}\ |H|\le k\}$. Thus, we have
$|V_k|=h^*(k,\setCount)=\sum_{i=0}^k{\setCount\choose i}$. Let
$h(k,\setCount)=k{\setCount\choose
k}$ if $k\le {\setCount\over 2}$, and $2^{\setCount}$ 
otherwise. Clearly, $h^*(k,\setCount)\le h(k,\setCount)$. By the
classical Stirling formula $k!\sim \sqrt{2\pi k}\cdot ({k\over
e})^k$, we have

\begin{eqnarray}
{\setCount\choose k}\le {\setCount^k\over k!}=\bigO(
({e\setCount\over k})^k).\label{stirling-ineqn}
\end{eqnarray}

Let $\gamma_{k,\setCount }$ be given as
Definition~\ref{shared-sample-def}. There are two cases to be
discussed.

Case 1: $1\le k\le {\setCount\over 2}$. By inequality
(\ref{stirling-ineqn}), we have
\begin{eqnarray}\log
{1\over \gamma_{k,\setCount}}&=&\log {\setCount
h^*(k,\setCount)\over \gamma}\le \log
{\setCount h(k,\setCount)\over \gamma}\\
&=&O(k\log({3 \setCount\over k})+\log k\setCount+\log {1\over \gamma})\\
&=&O(k\log({3 \setCount\over k})+\log (k^2\cdot {3\setCount\over k})+\log {1\over \gamma})\\
&=&O(k\log({3 \setCount\over k})+2\log k+\log{3\setCount\over k}+\log {1\over \gamma})\\
&=&O(k\log({3 \setCount\over k})+\log {1\over \gamma}).
\end{eqnarray}

Case 2: $n\ge k> {\setCount\over 2}$. It is trivial that
\begin{eqnarray}\log
{1\over \gamma_{k,\setCount}}&=&\log {\setCount
h^*(k,\setCount)\over \gamma}\le \log
{\setCount h(k,\setCount)\over \gamma}=\log {\setCount 2^\setCount\over \gamma}\\
&=&O(k\log({3 \setCount\over k})+\log {1\over \gamma}).
\end{eqnarray}

Thus, ${1\over \gamma_{k,\setCount}}=O(k\log({3 \setCount\over
k})+\log {1\over \gamma})$ for all $1\le k\le \setCount $. Thus,
\begin{eqnarray}
g(\epsilon, \gamma,k,\setCount)=R_D(\epsilon,\gamma_{k,\setCount
},k)&=&\bigO({1\over \epsilon^2}\log {1\over \gamma_{k,\setCount
}})=\bigO({1\over \epsilon^2}(k\log ({3\setCount\over k})+\log
{1\over \gamma})).
\end{eqnarray}

By Lemma~\ref{difference-lemma}, with $g(\epsilon,
\gamma,k,\setCount)=R_D(\epsilon,\gamma_{k,\setCount},k)$ random
samples from each set $A_i$ from each set $A_i$, the probability
that one of at most $\setCount h^*(k,\setCount)$ cases fails to
satisfy inequality (\ref{set2-diff-approx-condition}) is at most
$\setCount h^*(k,\setCount)\cdot \gamma_{k, \setCount}\le \gamma$ by
inequality~(\ref{prob-Add-Ineqn}).
\end{proof}

The random samples from each input set $A_i$ are collected in the
beginning of the algorithm of ApproximateMaximumCoverage(.), and are
stored in the list $R_i$. Virtual Function Implementation 2 gives
such an consideration.

\vskip 10pt

\addtocounter{problem-counter}{1}

 {\bf Algorithm~{\arabic{problem-counter} }: Virtual Function Implementation 2}

The parameters $L', A_i,s_i, R_i, \epsilon', \gamma, k,\setCount$
follow from those in Algorithm~{\arabic{problem-counter-max-cover}}.

\qquad RandomSamples$(A_i, R_i, \xi,\gamma,k,\setCount)$

\qquad \{

\qquad \qquad Generate a list $R_i$ of $g(\epsilon,
\gamma,k,\setCount)$ random samples of $A_i$;

\qquad\qquad Mark all elements of $R_i$ as white.

\qquad \}

\vskip 10pt

\qquad RandomTest($L', A_i, w$)

\qquad \{



\qquad \qquad Let $A_{t_j}$ be the newly picked set saved in $L'$
($L'=A_{t_1},A_{t_2},\cdots, A_{t_j}$);

\qquad \qquad For each  $a$ in the list $R_i$ of random samples from
$A_i$

\qquad \qquad \{

\qquad \qquad\qquad If $a\in A_{t_j}$ then mark $a$ as black in
$R_i$;

\qquad \qquad \}

\qquad\qquad Let $r_{i,j}$ be  the number of white items in $R_i$
(it may have multiplicity);

\qquad\qquad return  $r_{i,j}$;

\qquad \}



 \vskip 10pt

\qquad ApproximateSetDifferenceSize$(L', A_i, R_i,s_i, \epsilon',
\gamma, k,\setCount)$

\qquad \{

\qquad \qquad Let $w=g(\epsilon,\gamma, k,\setCount)$;

\qquad \qquad Let $r_{i,j}=$RandomTest($L', A_i, w$);

\qquad \qquad Return $s={r_{i,j}\over w}\cdot s_i$ ;

\qquad \}

\vskip 10pt

\qquad  ProcessSet($A_i)$ \{  \} (Do nothing to set $A_i$)

\vskip 10pt

{\bf End of Algorithm}

\vskip 10pt

Lemma~\ref{improve-ramdom-sample-lemma} shows approximation for
maximum coverage is possible via one round random samplings from
input sets. It shows how to control the number of random samples
from each input set to guarantee small failure probability of the
approximation algorithm. It slightly reduces the complexity by
multiple rounds of random samplings described in
Theorem~\ref{appr-cover-theorem}.

Lemma~\ref{improve-ramdom-sample-complexity-lemma} gives the
complexity for the algorithm with Virtual Function Implementation 2.

\begin{lemma}\label{improve-ramdom-sample-complexity-lemma}Assume that $\epsilon, \gamma\in (0,1)$ and
$k,\setCount\in\integerSet $. Let $k$ be an integer parameter and
$L$ be a list of sets $A_1,A_2,\cdots, A_{\setCount}$ for a maximum
coverage problem. Let $L'_*$ be the sublist $L'$ after  algorithm
ApproximateMaximumCover(.) is completed. Assume that the algorithm
uses Virtual Function Implementation 2
\begin{enumerate}[1.]
\item\label{part1-improve-ramdom-sample-complexity-lemma}
After the $j$-th iteration of the loop from
line~\ref{For-loop-start} to
line~\ref{call-app-diff-line-max-cover2} of
 the algorithm ApproximateMaximumCover(.), the returned
value $r_{i,j}$ from RandomTest($L', A_i, w$) is equal to
$\test(L_*'(j), R_i)$ (see Definition~\ref{test-function-def}).


\item
The algorithm with Virtual Function Implementation 2 has
complexity\\
$(T_2(\xi, \gamma,k, \setCount), R_2(\xi, \gamma,k, \setCount),
Q_2(\xi, \gamma,k, \setCount))$ such that for each set $A_i$, it
maintains $t_{i,j}=|R_i-U(L'_*(j))|$, where
\begin{eqnarray}
T_2(\xi, \gamma,k, \setCount)&=&\bigO({k^3\over \xi^2}(k\log
{3\setCount\over k}+\log
{1\over \gamma})\setCount),\\
R_2(\xi, \gamma,k, \setCount)&=&\bigO({k^2\over \xi^2}(k\log
{3\setCount\over k}+\log
{1\over \gamma})\setCount),  \ \ \ {\rm and}\\
Q_2(\xi, \gamma,k, \setCount)&=&\bigO({k^3\over \xi^2}(k\log
{3\setCount\over k}+\log {1\over \gamma})\setCount),
\end{eqnarray}
where
$R_i$ is the list of random samples from $A_i$ and defined in
Lemma~\ref{improve-ramdom-sample-lemma}.
\end{enumerate}
\end{lemma}

\begin{proof}Let $A_1,\cdots, A_{\setCount}$ be the input list of sets, and $k$ be the integer parameter in the input.
 Consider $g(\epsilon',\gamma, k,\setCount)$ random (white) samples for each set as
in Lemma~\ref{improve-ramdom-sample-lemma}, where $\epsilon'$ is
defined in algorithm ApproximateMaximumCover(.). After a set
$A_{t_j}$ is added to $L'$, all the random samples in $R_i$ will be
checked if they are from $A_{t_j}$. For each white random sample $x$
in $R_i$ for all $t_j\not=i$, change $x$ to black if $x\in A_{t_j}$.
Thus, it takes $kg(\epsilon',\gamma, k,\setCount) \setCount$ time.
Count the white samples left in $R_i$ and save it in the variable
$r_{i,j}$. A simple induction can prove
Part~\ref{part1-improve-ramdom-sample-complexity-lemma}. In the
beginning all elements in $R_i$ are white. Assume that after
$j$-iterations of the loop from line~\ref{For-loop-start} to
line~\ref{call-app-diff-line-max-cover2} of
 the algorithm ApproximateMaximumCover(.), the number of white elements of $R_i$ is
$\test(L_*'(j), R_i)$. After $(j+1)$ iterations, list $L'$ has $j+1$
sets and $A_{t_{j+1}}$ as the last added. Since each white random
sample of $R_i$ in $A_{t_{j+1}}$ is changed to black for all
$t_{j+1}\not=i$, we have $r_{i,j+1}=\test(L_*'(j+1), R_i)$. Thus,
the returned value $r_{i,j}$ from RandomTest($L', A_i, w$) is equal
to $\test(L_*'(j), R_i)$


By Lemma~\ref{improve-ramdom-sample-lemma}, the algorithm has
complexity ($T_2(\xi, \gamma,\setCount), Q_2(\xi,
\gamma,\setCount),R_2(\xi, \gamma,\setCount)$) with
\begin{eqnarray}
T_2(\xi, \gamma,k, \setCount)&=&kg(\epsilon',\gamma,
k,\setCount)\setCount=\bigO({k^3\over \xi^2}(k\log ({3\setCount\over
k})+\log
{1\over \gamma})\setCount),\\
R_2(\xi, \gamma,k, \setCount)&=&g(\epsilon',\gamma,
k,\setCount)\setCount=\bigO({k^2\over
\xi^2}(k\log ({3\setCount\over k})+\log {1\over \gamma})\setCount), \ \ \ {\rm and}\\
Q_2(\xi, \gamma,k, \setCount)&=&kg(\epsilon',\gamma, k,\setCount)
\setCount =\bigO({k^3\over \xi^2}(k\log ({3\setCount\over k})+\log
{1\over \gamma})\setCount).\\
\end{eqnarray}
\end{proof}

Theorem~\ref{appr2-cover-theorem} states that the improved
approximation algorithm for the maximum coverage problem has a
reduced complexity while keeping the same approximation ratio
$(1-{1\over e})$ for $((0,0),(0,0))$-list as input (see
Definition~\ref{app-list-def}). The algorithm is based on one round
samplings from the input sets. Now we give the  proof of
Theorem~\ref{appr2-cover-theorem}.

\vskip 10pt

\begin{proof}[Theorem~\ref{appr2-cover-theorem}]. Use $g(\epsilon',\gamma, k,\setCount)$ random samples for each
set $S_i$. It follows from Lemma~\ref{improve-ramdom-sample-lemma},
and Lemma~\ref{appr1-cover-lemma}. Select
$\xi=\min({\epsilon\eta\beta \over 4e^{\beta}k},\epsilon)$,
 where $\eta$ is
defined in Lemma~\ref{help-app-lemma} and $\beta$ is the same as
that in Lemma~\ref{appr1-cover-lemma}. With the condition
$\epsilon\in (0,1)$, the accuracy of approximation follows from
Lemma~\ref{help-app-lemma}.

By Lemma~\ref{improve-ramdom-sample-complexity-lemma}, its
complexity is  $(T(\epsilon,\gamma,k,\setCount), R(\epsilon,\gamma
,k,\setCount), Q(\epsilon,\gamma ,k,\setCount))$
\begin{eqnarray*}
T(\epsilon,\gamma,k,\setCount)&=&T_2(\xi,
\gamma,k,\setCount)=\bigO({k^5\over \epsilon^2}(k\log
({3\setCount\over k})+\log
{1\over \gamma})\setCount),\\
R(\epsilon,\gamma,k,\setCount)&=&R_2(\xi,
\gamma,k,\setCount)=\bigO({k^4\over \epsilon^2}(k\log
({3\setCount\over k})+\log
{1\over \gamma})\setCount), \\
Q(\epsilon,\gamma,k,\setCount)&=&Q_2(\xi,
\gamma,k,\setCount)=\bigO({k^5\over \epsilon^2}(k\log
({3\setCount\over k})+\log {1\over \gamma})\setCount).
\end{eqnarray*}
\end{proof}

We have Theorem~\ref{appr3-cover-theorem} that gives a slightly less
approximation ratio and has a less time complexity. The function
$(1-{1\over x})^x$ is increasing for all $x\in [2,+\infty)$ and
$\lim_{x\rightarrow +\infty}(1-{1\over x})^x={1\over e}$. This
implies that $\lim_{k\rightarrow+\infty}(1-{\beta\over
k})^k=\lim_{k\rightarrow+\infty}((1-{\beta\over
k})^{k\over\beta})^{\beta}={1\over e^{\beta}}$, and
$(1-(1-{\beta\over k})^k-\xi)>(1-{1\over e^{\beta}}-\xi)$.

\begin{theorem}\label{appr3-cover-theorem} Let $\rho$ be a constant in
$(0,1)$.
For  parameters $\xi,\gamma\in (0,1)$ and
$\alpha_L,\alpha_R,\delta_L,\delta_R\in [0,1-\rho]$, there is an
algorithm to give a $(1-(1-{\beta\over k})^k-\xi)$-approximation for
the maximum cover problem, such that given a
$((\alpha_L,\alpha_R),(\delta_L,\delta_R))$-list $L$ of finite sets
 $A_1,\cdots, A_{\setCount}$ and an integer $k$,  with probability at
least $1-\gamma$, it returns an integer $z$ and a subset $H\subseteq
\{1,2,\cdots, \setCount\}$ that satisfy
\begin{enumerate}[1.]
\item
$|\cup_{j\in H}A_j|\ge (1-(1-{\beta\over k})^k-\xi)C^*(L,k)$ and
$|H|=k$,
\item
$((1-\alpha_L)(1-\delta_L)-\xi)|\cup_{j\in H}A_j|\le z\le
((1+\alpha_R)(1+\delta_R)+\xi)|\cup_{j\in H}A_j|$, and
\item
Its complexity is  $(T(\xi,\gamma,k,\setCount), R(\xi,\gamma
,k,\setCount  ), Q(\xi,\gamma ,k,\setCount))$ with
\begin{eqnarray*}
T(\xi,\gamma,k,\setCount)&=&\bigO({k^3\over \xi^2}(k\log({3
\setCount\over k})+\log
{1\over \gamma})\setCount),\\
R(\xi,\gamma,k,\setCount)&=&\bigO({k^2\over \xi^2}(k\log ({3
\setCount\over k})+\log
{1\over \gamma})\setCount),\ \ {\rm and} \\
Q(\xi,\gamma,k,\setCount)&=&\bigO({k^3\over \xi^2}(k\log({3
\setCount\over k})+\log {1\over \gamma})\setCount),
\end{eqnarray*}
where $\beta={(1-\alpha_L)(1-\delta_L)\over
(1+\alpha_R)(1+\delta_R)}$.
\end{enumerate}
\end{theorem}

\begin{proof}  The accuracy of approximation follows from Lemma~\ref{improve-ramdom-sample-lemma},
and Lemma~\ref{appr1-cover-lemma}. By
Lemma~\ref{improve-ramdom-sample-complexity-lemma}, its complexity
is  $(T(\xi,\gamma,k,\setCount), R(\xi,\gamma ,k,\setCount),
Q(\xi,\gamma ,k,\setCount))$ with
\begin{eqnarray*}
T(\xi,\gamma,k,\setCount)&=&\bigO({k^3\over \xi^2}(k\log({3
\setCount\over k})+\log
{1\over \gamma})\setCount),\\
R(\xi,\gamma,k,\setCount)&=&\bigO({k^2\over \xi^2}(k\log ({3
\setCount\over k})+\log
{1\over \gamma})\setCount),\ \ {\rm and} \\
Q(\xi,\gamma,k,\setCount)&=&\bigO({k^3\over \xi^2}(k\log({3
\setCount\over k})+\log {1\over \gamma})\setCount).
\end{eqnarray*}
\end{proof}

Corollary~\ref{appr3-cover-corol} gives the case that we have exact
sizes for all input sets, and uniform random sampling for each of
them. Such an input is called $((0,0),(0,0))$-list according to
Definition~\ref{app-list-def}.

\begin{corollary}\label{appr3-cover-corol}
For parameters $\xi, \gamma\in (0,1)$, there is an algorithm to give
a $(1-(1-{1\over k})^k-\xi)$-approximation for the maximum  cover
problem, such that given a $((0,0),(0,0))$-list $L$ of finite sets
$A_1,\cdots, A_{\setCount}$ and an integer $k$,  with probability at
least $1-\gamma$, it returns an integer $z$ and a subset $H\subseteq
\{1,2,\cdots, \setCount\}$ that satisfy
\begin{enumerate}[1.]
\item
$|\cup_{j\in H}A_j|\ge (1-(1-{1\over k})^k-\xi)C^*(L,k)$ and
$|H|=k$,
\item
$(1-\xi)|\cup_{j\in H}A_j|\le z\le (1+\xi)|\cup_{j\in H}A_j|$, and
\item
Its complexity is  $(T(\xi,\gamma,k,\setCount),
R(\xi,\gamma,k,\setCount), Q(\xi,\gamma,k,\setCount))$ with
\begin{eqnarray*}
T(\xi,\gamma,k,\setCount)&=&\bigO({k^3\over \xi^2}(k\log ({3
\setCount\over k})+\log
{1\over \gamma})\setCount),\\
R(\xi,\gamma,k,\setCount)&=&\bigO({k^2\over \xi^2}(k\log ({3
\setCount\over k})+\log
{1\over \gamma})\setCount),\ \ {\rm and} \\
Q(\xi,\gamma,k,\setCount)&=&\bigO({k^3\over \xi^2}(k\log ({3
\setCount\over k})+\log {1\over \gamma})\setCount).
\end{eqnarray*}
\end{enumerate}
\end{corollary}

\section{Hardness of Maximum Coverage with Equal Size
Sets}\label{equali-size-section}

In this section, we show that the special case of maximum coverage
problem with equal size sets input is as hard as the general maximum
coverage problem. This gives a hard core for the maximum coverage
problem. When $A_1,\cdots, A_{\setCount}$ are of the same size
$\setSize$ , the input size is measured as $\setSize\setCount $.
Thus, the input size $\setSize\setCount $ for the maximum coverage
problem with equal set size is controlled by two independent
parameters $\setSize$ and $\setCount$. It helps us introduce the
notion of partial sublinear time computation at
Section~\ref{partial-sublinear-sec}.

The classical set cover problem is that given a set $U$ of elements
(called the universe) and a collection $S$ of sets whose union
equals the universe $U$, identify the smallest sub-collection of $S$
whose union equals the universe.

\begin{definition}
The {\it $s$-equal size maximum coverage} problem is the case of
maximum coverage problem when the input list of sets are all of the
same size $s$. The {\it equal size maximum coverage} problem is the
case of maximum coverage problem when the input list of sets are all
of the same size.
\end{definition}

\begin{theorem} Let $c$ be an positive real number and $s$ be an
integer parameter.
\begin{enumerate}
\item
There is a polynomial time reduction from a set cover problem with
set sizes bounded by $s$ to $s$-equal size maximum coverage problem.
\item
Assume there is a polynomial time $c$-approximation algorithm for
$s$-equal size maximum coverage problem, then there is a polynomial
time $(c-\littleo(1))$-approximation algorithm for the maximum
coverage problem with input sets $A_1,\cdots, A_{\setCount}$ of size
$|A_i|\le s$ for $i=1,2,\cdots,\setCount$ .
\end{enumerate}
\end{theorem}

\begin{proof} We use the following two cases to prove the two
statements in the theorem, respectively.
\begin{enumerate}
\item\label{first-case-proof}
Let $A_1,A_2,\cdots, A_{\setCount}$ be the input for a set cover
problem, and none of $A_1,A_2,\cdots, A_{\setCount}$ is empty set.
Without loss of generality, assume $t=|A_1|=
\max\{|A_1|,|A_2|,\cdots, |A_{\setCount}|\}$. Let $A_0$  be a new
set with $|A_0|=t$ and $A_0\cap (A_1\cup A_2\cup\cdots\cup
A_{\setCount})=\emptyset$. Construct a new list of sets $A_0, A_1',
A_2',\cdots, A_{\setCount}'$ such that each $A_i'=A_i\cup
A_0[t-|A_i|]$ for $i=1,2,\cdots, \setCount$, where $A_0[u]$ is the
first $u$ elements of $A_0$ (under an arbitrary order for the
elements in $A_0$). It is easy to see that $A_1, A_2,\cdots,
A_{\setCount}$ has a $k$ sets solution if and only if $A_0, A_1',
A_2',\cdots, A_{\setCount}'$ has a $k+1$ sets solution for the set
cover problem.
\item
Since maximum coverage problem has a polynomial time $(1-{1\over
e})$-approximation algorithm, We assume that $c$ is a fixed positive
real number. When $k$ is fixed, a brute force polynomial time
solution is possible to find an maximum union solution for the
maximum coverage problem. Therefore, we assume that ${1\over
k}=\littleo(1)$.  Let $A_1,A_2,\cdots, A_{\setCount}$ be an input
for a maximum coverage problem with a integer parameter $k$. Without
loss of generality, assume $|A_1|$ is the largest as
Case~\ref{first-case-proof}. Let $A_1^*=A_1$ and
$A_j^*=(A_1-A_j)[|A_0|-|A_j|]\cup A_j$ for $j=2,3,\cdots,
\setCount$. Consider the maximum coverage problem $A_1^*, A_2^*,
A_3^*,\cdots, A_{\setCount}^*$. Assume that the maximum coverage
problem $A_1^*, A_2^*, A_3^*,\cdots, A_{\setCount}^*$ has a
$c$-approximation $A_{i_1}^*, A_{i_2}^*,\cdots, A_{i_k}^*$.

1)$1\in \{i_1,i_2,\cdots, i_k\}$ ($A_1^*$ is one of the sets in the
solution). We have that $A_{i_1},\cdots, A_{i_k}$ is a
$c$-approximation for the maximum coverage problem for the input
$A_1,A_2,\cdots, A_{\setCount}$.

2) $1\not\in \{i_1,i_2,\cdots, i_k\}$. Let $A_{i_j}^*$ be that set
in the solution such that $|A_{i_j}^*-\cup_{v\not=j} A_{i_v}^*|$ is
the least. Clearly,  $|A_{i_j}^*-\cup_{v\not=j} A_{i_v}^*|\le
{|A_{i_1}^*\cup\cdots\cup A_{i_k}^*|\over k}$. Thus, $|A_1\cup
(\cup_{v\not=j}A_{i_v})|=|A_1^*\cup (\cup_{v\not=j}A_{i_v}^*)|\ge
(1-{1\over k})|A_{i_1}^*\cup A_{i_2}^*\cup \cdots \cup A_{i_k}^*|$.
Therefore, we have a $(c-\littleo(1))$-approximation $|A_1\cup
(\cup_{v\not=j}A_{i_v})|$ for the maximum coverage problem with the
input $A_1,A_2,\cdots, A_{\setCount}$.
\end{enumerate}
\end{proof}

Our partial sublinear time algorithm can be also applied to the
equal size maximum coverage problem, which has size
$\setCount\setSize$ controlled by two parameters $\setCount$ and
$\setSize$. Our algorithm has a time complexity independent of
$\setSize$ in the first model that gives $\bigO(1)$ time random
element, and $\bigO(1)$ answer for any membership query. Our partial
sublinear time approximation algorithm for the maximum coverage
problem becomes sublinear time algorithm when $\setSize\ge
\setCount^c$ for a fixed $c>0$.

\section{Inapproximability of Partial Sublinear Time Computation}\label{partial-sublinear-sec}

In this section, we introduce the concept of partial sublinear time
computation. The maximal coverage has a partial sublinear constant
factor approximation algorithm. On the other hand, we show that an
inapproximability result for equal size maximum coverage, which is
defined in Section~\ref{equali-size-section} and is a special case
of maximum coverage problem, if the time is
$q(\setSize)\setCount^{1-\epsilon}$, where $\setCount$  is the
number of sets.  This makes the notion of partial sublinear
computation different from conventional sublinear computation.

The inapproximability result is derived on a randomized
computational model that includes the one used in developing our
randomized algorithms for the maximum coverage problem. The
randomized model result needs some additional work than a
deterministic model to prove the inapproximability. A deterministic
algorithm with $q(\setSize)\setCount^{1-\epsilon}$ time let some set
be never queried by the computation, but all sets can be queried in
randomized algorithm with $q(\setSize)\setCount^{1-\epsilon}$ time
as there are super-polynomial (of both $\setSize$ and $\setCount$)
many paths.


\subsection{Model for Inapproximation}

We define a more generalized randomized computation model than that
given by Definition~\ref{input-list-def}. In the model given by
definition~\ref{model-def-lower-bound}, it allows to fetch the
$j$-th element from an input set $A_i$. As we study a sublinear time
computation, its model is defined by
Definition~\ref{model-def-lower-bound}. It is more powerful than
that given in Definition~\ref{input-list-def}. The inapproximation
result derived in this model also implies a similar result in the
model of Definition~\ref{input-list-def}.


\begin{definition}\label{model-def-lower-bound}
A randomized computation $T(.,.)$ for the maximum coverage problem
is a tree $T$ that takes  an input $k$ of integer and a list of
finite sets defined in Definition~\ref{input-list-def}.
\begin{enumerate}[1.]
\item
Each node of $T((L,k),.)$ (with input list $L$ of sets and integer
$k$ for the maximum coverage) allows any operation in
Definition~\ref{input-list-def}.
\item
A fetching statement $x=A_i[j]$ ($1\le j\le s_i= |A_i|$) lets $x$
get the $a_{i,j}\in A_i$, where set $A_i$ contains the elements
$a_{i,1},\cdots a_{i, s_i}$ (which may be unsorted).
\item
A {\it branching point} of $p$ that has $s$ children
$p_1,p_2,\cdots, p_s$ and is caused  by the following two cases

\begin{itemize}
\item
 $\randomElm(A_i)$ returns  a random element in $A_i=\{a_{i,1},\cdots, a_{i,s_i}\}$ such
that $p_j$ is the case that $a_{i,j}$ is selected in $A_i$, and
$s=s_i$.
\item
$\randomNum(s)$ returns  a random element in $\{0,1,\cdots, s-1\}$
for an integer $s>0$ such that $p_j$ is the case that $j+1$ is
returned.
\end{itemize}

\item
A computation {\it path} is determined by a series of numbers $r_0,
r_1, r_2,\cdots, r_t$ such that the $r_j$ corresponds to the $j$-th
branching point for $j=1,2,\cdots,t-1$, and $r_0$ is the root, and
$r_t$ is a leaf.
\item
A {\it partial path} $p$ is an initial part of a path that starts
from the root $r_0$ of computation to a node $q$.
\item\label{root-weight}
The root node $r_0$ has weight $w(r_0)=1$.
\item\label{partial-extension-weight}
If a partial path $p$ from root $r_0$ to a node $q$ that has
children $p_1,\cdots p_s$, and weight $w(q)$. Then
$w(p_1)=w(p_2)=\cdots =w(p_s)={w(q)\over s}$, where $w(p_i)$ is the
weight for $p_i$.
\item
The weight of a path from the root $r_0$ to a leaf $q$ has the
weight $w(q)$, which is the weight of $q$.
\item
The output of the randomized computation $T((L,k),.)$ (with input
$(L,k)$) on a  path $p$ is defined to be $T((L,k),p)$.
\end{enumerate}
\end{definition}

The weight function $w(.)$ determines the probability of a partial
path or path is generated in the randomized computation.
In Definition~\ref{def-shared-path}, we give the concept of  a
shared path for randomized computation under two different inputs of
lists of sets. Intuitively, the computation of the two shared paths
with different inputs does not have any difference, gives both the
same output, and has the same weight.

\begin{definition}\label{def-shared-path}
\scrod
\begin{itemize}
\item
Let $L$ be a list of sets $A_1,\cdots, A_{\setCount}$, and $L'$ be
another input list of $\setCount$ sets $A_1',\cdots,
A_{\setCount}'$. If $|A_i|=|A_i'|$ for $i=1,2,\cdots, \setCount$,
then $L$ and $L'$ are called {\it equal size} list of sets.
\item
Let $L$ be a list of sets $A_1,\cdots, A_{\setCount}$ and $L'$ an
another input of $\setCount$ sets $A_1',\cdots, A_{\setCount}'$ such
that they are of equal size. A partial path $p$ is {\it shared} by
$T((L,k),.)$ and $T((L',k),.)$ if
\begin{itemize}
\item
path $p$ gets the same result for $Query(x, A_i)$ and
$Query(x,A_i')$ for all queries along $p$,
\item
path $p$ gets the same result for fetching between $x=A_i[j]$ and
$x=A_i'[j]$,
\item
path $p$ gets the same result for each random access to
$\randomElm(A_i)$ and $\randomElm(A_i')$, and
\item
path $p$ gets the same result for each random access to
$\randomNum(s)$.
\end{itemize}
\end{itemize}
\end{definition}


\begin{definition}Let $T((L,k),.)$ be a randomized computation for the maximum coverage with
an input list $L$ and an integer $k$.
\begin{itemize}
\item
 Define $P(1)$ to be the set that contains the
partial path with the root.
\item
If $p\in P(a)$ and $p$ is from root $r_0$ to a branching point $q$
with children $q_1,\cdots, q_t$, then each partial path from $r_0$
to $q_i$ belongs to  $P(a+1)$ for $i=1,2,\cdots, t$.
\item
$P(a+1)$ contains all paths (from the root to leave) of length at
most $a+1$ nodes.
\end{itemize}
\end{definition}

\begin{lemma}\label{path-weight-sum-lemma}
Let $T((L,k),.)$ be a randomized computation for the maximum coverage with
an input $(L,k)$.
\begin{enumerate}
\item\label{case1-lemma}
$\sum_{p\in P(a)}w(p)=1$ for all $a\ge 1$.
\item\label{case2-lemma}
$\sum_{path\ p}w(p)=1$.
\end{enumerate}
\end{lemma}

\begin{proof}
It can be proven via an induction. It is true for $a=1$ by
definition. Assume $\sum_{p\in P(a)}w(p)=1$. We have $\sum_{p\in
P(a+1)}w(p)=1$ by item~\ref{partial-extension-weight} of
Definition~\ref{model-def-lower-bound}. Statement~\ref{case2-lemma}
follows from Statement~\ref{case1-lemma}.
\end{proof}

\begin{definition}\label{appr-def}
A {\it partial sublinear} $\bigO(t_1(\setCount)t_2(\setSize))$ time
$(u(\setCount ,\setSize), v(\setCount ,\setSize))$-approximation
algorithm $T(.,.)$ for the maximum coverage problem if it satisfies
the following conditions:
\begin{enumerate}[1.]
\item
It runs $\bigO(t_1(\setCount)t_2(\setSize))$ steps along every path.
\item
The two functions have $t_1(\setCount)=\littleo(\setCount)$ or
$t_2(\setSize)=\littleo(\setSize)$, and
\item
The sum of weights $w(p)$ of paths $p$ that satisfy ${C^*(L,k)\over
u(\setCount ,\setSize)}-v(\setCount ,\setSize)\le T((L,k),p)\le
u(\setCount ,\setSize)\cdot C^*(L,k)+v(\setCount ,\setSize)$ is  at
least ${3\over 4}$, where $\setSize=\max\{|A_i|:i=1,\cdots,
\setCount\}$.
\end{enumerate}
\end{definition}

\subsection{Inapproximation for Equal Size Maximum Coverage}

We derive an inapproximability result for the equal size maximum
coverage problem in partial sublinear
$p(\setSize)\setCount^{1-\epsilon}$ time model. It contrasts the
partial sublinear time $\bigO(\setCount\poly(k))$ constant factor
approximation  for it.

It will be proven by contradiction. Assuming that there exists a
constant factor $q(\setSize)\setCount^{1-\epsilon}$ time randomized
algorithm for the maximum coverage. We construct two lists
$L:A_1,A_2,\cdots, A_{\setCount}$ and $L':A_1',A_2',\cdots,
A_{\setCount}'$ of sets that have the same size. In list $L$, we let
$A_1=A_2=\cdots =A_n$. All input sets in both $L$ and $L'$ are of
the same size. Two inputs for $L$ and $L'$ with the same parameter
$k$ for the maximum coverage will share most of their paths. There
are small number of sets $A_{i_1},A_{i_2},\cdots, A_{i_d}$ in $L$
such that they are queried by a small percent of paths (with sum of
their weights less than $1\%$). In list $L'$ we let $A_i'=A_i'$ for
all the sets except $A_{i_1}',\cdots, A_{i_d}'$, and let each
$A_{i_s}'\cap A_{i_t}'=\emptyset$. There is a large difference for
their maximum coverage solutions with the same parameter $k$ if $d$
is reasonably large.  This is possible as the time is controlled by
$q(\setSize)\setCount^{1-\epsilon}$ in each path. In other words,
$C^*(L,k)=|A_1|$ and $C^*(L,k)=d|A_1'|=d|A_1|$. There is an
approximate value from a shared path to be close to $C^*(L,k)$ and
$C^*(L',k)$ of two maximum coverage problems, respectively. It
derives a contradiction.

\begin{theorem}\label{separate-theorem}For nondecreasing functions
$t_1(\setCount), t_2(\setSize), v(\setSize): \integerSet\rightarrow
\realSet^+$ with $v(\setSize)=\littleo(\setSize)$ and
$t_1(\setCount)=\littleo(\setCount)$, the function $C^*(L,k)$ for
the equal size maximum coverage problem has no partial sublinear
$t_1(\setCount)t_2(\setSize)$ time $(u, v(\setSize))$-approximation
for any fixed $u>0$.
\end{theorem}

\begin{proof}It is proven by contradiction. Let $u$ be a fixed positive integer. Assume that the
maximum coverage problem $C^*(L,k)$ has a partial sublinear
$t_1(\setCount)t_2(\setSize)$ time $(u,v(\setSize))$-approximation
by a randomized computation $T(.,.)$.

Let
\begin{eqnarray}
c_1&=&200, \label{c1-set-eqn}{\rm \ \ \ and}\\
k&=&d=c_1u^2.\label{d-set-eqn}
\end{eqnarray}

Since $v(\setSize)=\littleo(\setSize)$, select $\setSize$  to be an
integer such that
\begin{eqnarray}
c_1u\cdot v(\setSize)&<& \setSize, \label{m2-select-ineqn}{\rm \ \ \
and}\\
\setSize &=&0({\rm mod}\ d).
\end{eqnarray}
 Select an integer $\setCount$  to be large enough such that
 \begin{eqnarray}
\max(2d, c_1\cdot d\cdot t_1(\setCount)t_2(\setSize))&<&
n\label{n2-select-ineqn}.
\end{eqnarray}

Let $\integerSet_\setSize$ be the set of integers $\{1,2,\cdots,
\setSize\}$. Let sets $A_1=A_2=\cdots = A_{\setCount}=\{1,2,\cdots,
{\setSize\over  d}\}$. Let $L$ be the list of sets $A_1, A_2,\cdots
, A_{\setCount}$.

For each $A_i$, define $Q(A_i)=\sum_{path\ p \ in \ T((L,k),.)\
queries\ A_i} w(p)$. If there are more than ${\setCount\over 2}$
sets $A_i$ with $Q(A_i)> {0.01\over d}$, then
\begin{eqnarray}
\sum_{i=1}^n Q(A_i)> {0.01\setCount\over 2d}.\label{sum-Q-Ai-ineqn}
\end{eqnarray}
For a path $p$, define $H(p)$ to be the number of sets $A_i$ queried
by $p$. Clearly, $H(p)\le t_1(\setCount)t_2(\setSize)$ since each
path runs in at most $t_1(\setCount)t_2(\setSize)$ steps.  We have
\begin{eqnarray*}
\sum_{i=1}^n Q(A_i)&\le& \sum_{p} w(p)H(p)\\
&\le&
t_1(\setCount)t_2(\setSize)\sum_{p}w(p)\\
&=&t_1(\setCount)t_2(\setSize).\ \ \ \ ({\rm by\
Lemma~\ref{path-weight-sum-lemma}}) \label{sum2-Q-Ai-ineqn}
\end{eqnarray*}

By inequalities (\ref{sum-Q-Ai-ineqn}) and (\ref{sum2-Q-Ai-ineqn}),
we have ${0.01\setCount\over 2d}< t_1(\setCount)t_2(\setSize)$,
which implies $n<200\cdot d\cdot t_1(\setCount)t_2(\setSize)$. This
contradicts inequality (\ref{n2-select-ineqn}). Therefore, there are
at least ${\setCount\over 2}$ sets $A_j$ with $Q(A_j)< {0.01\over
d}$. Let $J$ be the set $\{j:Q(A_j)< {0.01\over d}\}$. We have
$|J|\ge {\setCount\over 2}\ge d$ (by inequality
(\ref{n2-select-ineqn})).


Let $i_1<i_2<\cdots <i_d$ be the first $d$ elements in set $J$.
Define the list $L'$ of sets $A_1', A_2',\cdots, A_{\setCount}'$
with
\begin{eqnarray*}
A_{i_1}'&=&\{1,2,\cdots, {\setSize\over  d}\},\\
A_{i_2}'&=&\{{\setSize\over  d}+1,{\setSize\over  d}+2,\cdots, {2\setSize\over  d}\},\\
\cdots &&\cdots\\
A_{i_d}'&=&\{{(d-1)\setSize\over  d}+1,{(d-1)\setSize\over  d}+2,\cdots, \setSize\},\ \ {\rm and}\\
A_j'&=&A_j \ \ \ {\rm for\  every\ } j\in \{1,2,\cdots,
\setCount\}-\{i_1,i_2,\cdots, i_d\}.
\end{eqnarray*}

Clearly, the two lists $L$ and $L'$ are of the same size, and we
have
\begin{eqnarray}
C^*(L',k)&=&\setSize,\ \ \ \ {\rm and}\label{first-union-size-eqn}\\
C^*(L,k)&=&{\setSize\over  d}. \label{second-union-size-eqn}
\end{eqnarray}



A shared path has the same weight in both $T((L,k),.)$ and
$T(L',.)$. We have $\{p: $ $p$ is shared by $T((L,k),.)$ and
$T(L',.)\}$ has weight of at least $0.99$ in total. This is because
the difference between list $L$ and list $L'$ is between the sets
with indices in $\{i_1,i_2,\cdots, i_d\}$. The sum of weights $w(p)$
from the paths $p$ in $T((L,k),.)$ that $p$  queries $A_j$ with
$j\in \{i_1,i_2,\cdots, i_d\}$ is at most $d\cdot {0.01\over
d}=0.01$ since $Q(A_j)<{0.01\over d}$ for each $j\in J$.

There exists a $z$, which is equal to $T((L,k),p)=T((L',k),p)$ for
some shared path $p$, such that $z$ is a
$(u,v(\setSize))$-approximation for both $C^*(L',k)$ and $C^*(L,k)$.
Therefore,
\begin{eqnarray}
{C^*(L,k)\over u }-v(\setSize)&\le& z\le
u \cdot C^*(L,k)+v(\setSize),{\rm \ \ \ and}\\
{C^*(L',k)\over u }-v(\setSize)&\le& z\le u \cdot
C^*(L',k)+v(\setSize).
\end{eqnarray}

Therefore,
\begin{eqnarray}
{C^*(L',k)\over u }-v(\setSize)&\le& z\le u \cdot
C^*(L,k)+v(\setSize).
\end{eqnarray}

By equations (\ref{first-union-size-eqn}) and
(\ref{second-union-size-eqn}),
\begin{eqnarray}
{\setSize\over  u }-v(\setSize)&\le& {u \setSize\over
d}+v(\setSize).
\end{eqnarray}

Therefore,
\begin{eqnarray}
 {\setSize(d-u ^2)\over du }\le 2v(\setSize).\label{d-u-v-ineqn}
\end{eqnarray}

By inequality (\ref{c1-set-eqn}), and equation (\ref{d-set-eqn}),
and inequality (\ref{d-u-v-ineqn}), we have
\begin{eqnarray}
{\setSize\over  2u }\le {\setSize(d-{d\over 2})\over du }\le
{\setSize(d-u ^2)\over du }\le 2v(\setSize).
\end{eqnarray}

Therefore,
\begin{eqnarray}
\setSize\le  4u\cdot v(\setSize).
\end{eqnarray}

This brings a contradiction by inequality (\ref{m2-select-ineqn}).

\end{proof}

 Theorem~\ref{separate-theorem} implies there is no
 $\bigO((\setCount \setSize)^{1-\epsilon})$ time approximation for the maximum coverage
 problem. Thus, Theorem~\ref{separate-theorem} gives a natural example that has
partial sublinear time constant factor approximation, but has no
sublinear time approximation.

\begin{corollary}For nondecreasing functions
$t_1(\setCount), t_2(\setSize), v(\setSize): \integerSet\rightarrow
\realSet^+$ with $v(\setSize)=\littleo(\setSize)$ and
$t_1(\setCount)=\littleo(\setCount)$, the function the equal size
maximum coverage problem has no partial sublinear
$t_1(\setCount)t_2(\setSize)$ time $(u, v(\setSize))$-approximation
for any fixed $u>0$.
\end{corollary}

A lot of computational problems can be represented by a function on
a list of sets. For example, any bipartite graph $G(V_1, V_2,E)$ can
be represented by a list of sets $A_1,A_2,\cdots, A_{\setCount}$,
where $\setCount =|V_1|$ and $V_1=\{v_1,v_2,\cdots,
v_{\setCount}\}$, and each $A_i$ is a subset of $V_2$ with
$\setSize=|V_2|$ such that each one of $A_i$ has an edge adjacent to
$v_i\in V_1$. If we define $f(A_1,A_2,\cdots, A_{\setCount})$
 to be length of longest path in $G(V_1, V_2,E)$, the function $f(.)$ is a NP-hard.
If $g(A_1,A_2,\cdots, A_{\setCount})$ is defined to
 be number of paths of longest paths in $G(V_1, V_2,E)$, the function $g(.)$ is a \#P-hard.

\section{Maximum Coverage on Concrete
Models}\label{concrete-models-section}


In this section, we show some data structures that can support the
efficient implementation of the algorithm. We define the time
complexity of a randomized algorithm in our computation model. We
only use type 0 model for maximum coverage (see
Definition~\ref{input-list-def}).

\begin{definition}\label{app-epsilon-def} For  parameters $\epsilon,c\in (0,1)$ an
algorithm to give a $(c, \epsilon)$-approximation for the maximum
cover problem in type 0 model, such that given a list $L$ of finite
sets $A_1,\cdots, A_{\setCount}$ and an integer $k$, it returns an
integer $z$ and a subset $H\subseteq \{1,2,\cdots, \setCount\}$ that
satisfy
\begin{enumerate}[1.]
\item
$|\cup_{j\in H}A_j|\ge c C^*(L,k)$ and $|H|=k$, and
\item
$(1-\epsilon)|\cup_{j\in H}A_j|\le z\le (1+\epsilon)|\cup_{j\in
H}A_j|$.
\end{enumerate}
\end{definition}

\begin{definition}\label{complexity-def}
Assume that the complexity for getting one random sample from set
$A_i$ is $r(|A_i|)$ and the complexity for making one membership
query for set $A_i$ is $q(|A_i|)$. If an algorithm has a complexity
$(T(\epsilon, \gamma, k,\setCount), R(\epsilon, \gamma,
k,\setCount), Q(\epsilon, \gamma, k,\setCount))$, define its time
complexity by $T(\epsilon, \gamma, k,\setCount)+R(\epsilon, \gamma,
k,\setCount)r(\setSize)+Q(\epsilon, \gamma,
k,\setCount)q(\setSize)$.
\end{definition}

Theorem~\ref{appr2-cover-theorem} can be restated in the following
format. It will be transformed into several version based on the
data structures to store the input sets for the maximum coverage.

\begin{theorem}\label{main2-theorem}Assume that each set of size $\setSize$  can
generate a random element in $r(\setSize)$ time and answer a
membership query in $q(\setSize)$ time. Then there is a randomized
algorithm such that with probability at most $\gamma$,
ApproximateMaximumCover($A_1,A_2,\cdots, A_{\setCount}, \epsilon,
\gamma$) does not give an  $((1-{1\over e}),
\epsilon)$-approximation for the maximum coverage in type 0 model.
Furthermore, its complexity is
  $T(\epsilon,
\gamma, k, \setCount)+R(\epsilon, \gamma, k,
\setCount)r(\setSize)+Q(\epsilon, \gamma, k, \setCount)q(\setSize)$,
where $T(.), R(.)$, and $Q(.)$ are the same as those in
Theorem~\ref{appr2-cover-theorem}.
\end{theorem}

\begin{proof} It follows from
Definition~\ref{complexity-def},
Lemma~\ref{improve-ramdom-sample-lemma} and
Theorem~\ref{appr2-cover-theorem}.
\end{proof}

\subsection{Maximum Coverage on High Dimension Space}\label{Apply-high-dimension}

In this section, we apply the randomized algorithm for high
dimensional maximum  coverage problem. It gives an application to a
$\#$P-hard problem. The geometric maximum coverage with 2D
rectangles problem was studied and proven to be NP-hard in
~\cite{MegiddoSupowit84}. An approximation scheme for the 2D maximum
coverage with rectangles was developed in~\cite{LiWangZhangZhang15}
with time $\bigO({n\over \epsilon}\log ({1\over \epsilon})+m({1\over
\epsilon})^{\bigO(\min(\sqrt{m},{1\over\epsilon})})$.

\begin{definition}
An axis aligned rectangular shape $R$ is called integer rectangle if
all of its corners are lattice points. A special integer rectangle
is a called {\it 0-1-rectangle } if each corner  $(x_1,x_2,\cdots,
x_n)$ has $x_i\in\{0,1\}$ for $i=1,2, \cdots, \setCount$.
\end{definition}

{\it Geometric Integer Rectangular Maximum Coverage Problem:} Given
a list of integer rectangles $R_1, R_2, \cdots, R_{\setCount}$ and
integer parameter $k$, find $k$ of them $R_{i_1},\cdots, R_{i_k}$
that has the largest number of lattice points. The 0-1 Rectangle
Maximum Coverage problem is the Geometric Integer Rectangular
Maximum Coverage Problem with each rectangle to be 0-1 rectangle.

 This problem is
\#P-hard even at the special case $k=\setCount$ for counting the
total number of lattice points in the union of the $\setCount$
rectangles. A logical formula is considered to be in DNF if and only
if it is a disjunction of one or more conjunctions of one or more
literals. Counting the number of assignments to make a DNF true is
\#P-hard~\cite{ValiantSharp}.

\begin{proposition}
The 0-1 Rectangle Maximum Coverage problem is \#P-hard.
\end{proposition}

\begin{proof}
The reduction is from \#DNF, which is \#P-hard~\cite{ValiantSharp},
to the 0-1 Rectangle Maximum Coverage problem. For each conjunction
of literals $x_1^*x_2^*\cdots x_k^*$ (each $x_i^*\in
\{x_i,\overline{x_i}\}$ is a literal), all of the satisfiable
assignments to this term form the corners of a 0-1 rectangle.
\end{proof}

\begin{theorem}For parameters $\epsilon, \gamma\in (0,1)$, there is a randomized  an  $((1-{1\over e}),
\epsilon)$-approximation algorithm for the $d$-dimensional Geometric
Integer Rectangular Maximum Coverage Problem in type 0 model for,
and has time complexity $\bigO(T(\epsilon, \gamma,k,
\setCount)+R(\epsilon, \gamma,k,\setCount)d+Q(\epsilon, \gamma,
k,\setCount)d)$, where $T(.), R(.)$, and $Q(.)$ are the same as
those in Theorem~\ref{appr2-cover-theorem}. Furthermore, the failure
probability is at most $\gamma$.
\end{theorem}

\begin{proof}
For each integer rectangular shape $R_i$, we can find the number of
lattice points in $R_i$. It takes $\bigO(d)$ time to generate a
random lattice at a $d$-dimensional rectangle, and answer a
membership query to an input rectangle. There is uniform random
generator for the set of lattice points in $R_i$. The input list of
sets is a $((0,0),(0,0))$-list as there is a perfect uniform random
sampling, and has the exact number of lattice points for each set.
 It follows from
Theorem~\ref{main2-theorem}.
\end{proof}

\subsection{Maximum Coverage with Sorted List}\label{sorted-list-section}

In this section, we discuss that each input set for the maximum
coverage problem is in a sorted array. We have the
Theorem~\ref{sorted-theorem}.

\begin{theorem}\label{sorted-theorem} Assume each input set is in a sorted list. Then
with probability at least $1-\gamma$,
ApproximateMaximumCover($A_1,A_2,\cdots, A_{\setCount}, \epsilon,
\gamma$) outputs a $(1-{1\over e})$-approximation in time
$\bigO(T(\epsilon, \gamma, k,\setCount)+R(\epsilon,
\gamma,k,\setCount)+Q(\epsilon, \gamma,k, \setCount)\log \setSize)$,
where $T(.), R(.)$, and $Q(.)$ are the same as those in
Theorem~\ref{appr2-cover-theorem}, parameters $\epsilon,\gamma\in
(0,1)$ and $\setSize=\max\{|A_1|,|A_2|,\cdots,|A_{\setCount}|\}$.
\end{theorem}

\begin{proof} The input list of sets is a $((0,0),(0,0))$-list as a sorted list provides perfect uniform random sampling, and has
the exact number of items for each set. It takes $\bigO(1)$ steps to
get a random sample from each input set $A_i$, and $\bigO(\log
\setSize)$ steps to check membership.
 It follows from Lemma~\ref{improve-ramdom-sample-lemma} and
Theorem~\ref{main2-theorem}.
\end{proof}

\subsection{Maximum Coverage with Input as Unsorted Arrays}

In this section, we show our approximation when each input set is an
unsorted array of elements. A straightforward method is to sort each
set or generate a B-tree for each set. This would take
$\bigO(\setCount m\log \setSize)$ time, where $\setSize$  is the
size of the largest set.

The following implementations will be used to support the case when
each input set is unsorted. Whenever a set is selected to the
solution, it will be sorted so that it will be efficient to test if
other random samples from the other sets belong to it.

\vskip 10pt

\addtocounter{problem-counter}{1}

 {\bf Algorithm~{\arabic{problem-counter} }: Virtual Function Implementation 3}

The parameters $L', A_i,s_i, R_i, \epsilon', \gamma, k,\setCount$
follow from those in Algorithm~{\arabic{problem-counter-max-cover}}.

\qquad RandomSamples$(A_i, R_i, \xi,\gamma,k,\setCount)$

\qquad \{

\qquad \qquad    The same as that in Virtual Function Implementation
2;

\qquad \}

\vskip 10pt

\qquad RandomTest(.)\{ the same as that in Virtual Function
Implementation 2;\qquad \}

\vskip 10pt

\qquad ApproximateSetDifferenceSize$(L', A_i, R_i, s_i, R_i,
\epsilon', \gamma, k,\setCount)$

\qquad \{   The same as that in Virtual Function Implementation 2;
\qquad \}

\vskip 10pt

\qquad  ProcessSet($A_i)$

\qquad \{

\qquad \qquad Sort $A_i$;

\qquad \}

\vskip 10pt

{\bf End of Algorithm}

\vskip 10pt

\begin{lemma}\label{improve-ramdom-sample-unsorted-lemma}Assume that $\epsilon, \gamma\in (0,1)$ and
$k,\setCount\in\integerSet $. The algorithm can be implemented with
sort merge for the function Merge(.) in complexity ($T_3(\xi,
\gamma,k,\setCount), Q_3(\xi, \gamma,k,\setCount),R_3(\xi,
\gamma,k,\setCount)$) such that for each set $A_i$, it maintains
$t_{i,j}=|R_i-U(L'(j))|$, where
\begin{eqnarray}
T_3(\xi, \gamma,k, \setCount)&=&\bigO(kg(\epsilon',\gamma, k,\setCount)\setCount+k\setSize(\log k+\log \setSize)),\\
Q_3(\xi, \gamma,k, \setCount)&=&kg(\epsilon',\gamma, k,\setCount) \setCount , \\
R_3(\xi, \gamma,k, \setCount)&=&g(\epsilon',\gamma, k,\setCount)\setCount, \ \ \ {\rm and}\\
\end{eqnarray}
 $\epsilon'$ is defined in algorithm ApproximateMaximumCover(.),
and $R_i$ is the set of random samples from $A_i$ and defined in
Lemma~\ref{improve-ramdom-sample-lemma}.
\end{lemma}

\begin{proof}Let $A_1,\cdots, A_{\setCount}$ be the input list of sets, and $k$ be the integer parameter in the input.

Consider $g(\epsilon',\gamma, k,\setCount)$ random samples for each
set as in Lemma~\ref{improve-ramdom-sample-lemma}, where $\epsilon'$
is defined in algorithm ApproximateMaximumCover(.). After a set
$A_i$ is selected, it takes $\bigO(\setSize\log \setSize)$ to sort
the elements in newly selected set, and adjust the random samples
according to ApproximateSetDifferenceSize(.) in Virtual Function
Implementation 3.

It let $R_j$ become $R_j-A_i$ for all $j\not=i$. Thus, it takes
$kg(\epsilon',\gamma, k,\setCount) \setCount $ time.
Mark those samples that are in the selected sets, and count the
samples left (unmarked). It is similar to update $R_j$ and $t_{i,j}$
as in the proof of
Lemma~\ref{improve-ramdom-sample-complexity-lemma}.
 The approximation $t_{i,j}$ to
$|A_i-U(L'(j))|$ can be computed as ${t_{i,j}\over w}\cdot s_i$
($w=g(\epsilon,\gamma, k,\setCount)$) as in line
\ref{alg-approx-diff-return-line} in ApproximateDifference(.).

\end{proof}

\begin{theorem}Assume each input set is an unsorted list. Then
with probability at least $1-\gamma$,
ApproximateMaximumCover($A_1,A_2,\cdots, A_{\setCount}, \epsilon,
\gamma$) outputs a an  $((1-{1\over e}), \epsilon)$-approximation
for the maximum coverage in type 0 model. Its time complexity is
$\bigO({k^5\over \epsilon^2}(k\log({3 \setCount\over k})+\log
{1\over \gamma})\setCount)+k\setSize(\log k+\log \setSize)))$, where
$\epsilon,\gamma\in (0,1)$ and
$\setSize=\max\{|A_1|,|A_2|,\cdots,|A_{\setCount}|\}$.
\end{theorem}

\begin{proof}At line~\ref{merege-line} in ApproximateMaximumCover(.), we build up
a B-tree to save all the elements in the sets that have been
collected to $L'$. The total amount time to build up $L'$ in the
entire algorithm is $\bigO(k\setSize(\log k+\log \setSize))$.

Use $g(\epsilon',\gamma, k,\setCount)$ random samples for each set
$S_i$. It follows from Lemma~\ref{improve-ramdom-sample-lemma},
Lemma~\ref{appr1-cover-lemma}, and
Lemma~\ref{improve-ramdom-sample-unsorted-lemma}. Select
$\xi=\min({\epsilon\eta\beta \over 4e^{\beta}k},\epsilon)$,
 where $\eta$ is
defined in Lemma~\ref{help-app-lemma} and $\beta$ is the same as
that in Lemma~\ref{appr1-cover-lemma}. With the condition
$\epsilon\in (0,1)$, the accuracy of approximation follows from
Lemma~\ref{help-app-lemma}.
\end{proof}

\subsection{Maximum Coverage with B-Tree}\label{online-section}

In this section, we discuss an online model. The sorted array
approach is not suitable for the online model as insertion or
deletion may take $\Omega(\setSize)$ steps in the worst case.
Therefore, we discuss the following model.

The B-tree implementation is suitable for dynamic sets that can
support insertion and deletion for their members. It is widely used
in database as a fundamental data structure. It takes $\bigO(\log
\setSize)$ time for query, insertion, or deletion. B-tree can be
found in a standard text book of algorithm (for example,
\cite{CormenLeisersonRivestStein01}).

We slightly revise the B-tree structure for its application to the
maximum coverage problem.  Each set for the maximum coverage problem
is represented by a B-tree that saves all data in the leaves. Let
each internal node $t$ contain a number $C(t)$ for the number of
leaves under it. Each element has a single leaf in the B-tree (we do
not allow multiple copies to be inserted in a B-tree for one
element). We also let each node contain the largest values from the
subtrees with roots at its children.  It takes $\bigO(\log
\setSize)$ time to generate a random element, and $\bigO(\log
\setSize)$ time to check if an element belongs to a set.

\begin{definition}
For two real numbers $a$ and $b$, define N$[a,b]$ to be the set of
integers $x$ in $[a,b]$.
\end{definition}

\vskip 10pt

\addtocounter{problem-counter}{1}

 {\bf Algorithm~{\arabic{problem-counter} }:} Rand($T, t)$

Input a B-tree $T$ with a node $t$ in $T$.

{\bf Steps:}

\begin{enumerate}[1.]

\item \qquad If $t$ is a leaf, return $t$;

\item
\qquad Let $t_1,\cdots, t_k$ be the children of $t$;

\item
\qquad Select a random integer $i\in N[1,C(t_1)+\cdots +C(t_k)]$;

\item
\qquad Partition N$[1,C(t_1)+\cdots +C(t_k)]$ into

\qquad $I_1=$N$[1, C(t_1)], I_2=$N$[C(t_1)+1,
C(t_1)+C(t_2)],\cdots,$ and

\qquad $I_k=$N$[C(t_1)+\cdots+C(t_{k-1})+1, C(t_1)+\cdots+C(t_{k}]$;

\item

\qquad Find $I_j$ with $i\in I_j$ and return Rand($T, t_j)$;

\end{enumerate}

{\bf End of Algorithm}

\vskip 10pt

\begin{lemma}\label{random-b-tree-lemma}
There is a B-tree implementation such that it can generate a random
element in $\bigO(\log \setSize)$ time, where $\setSize$  is the
number of elements saved in the B-tree.
\end{lemma}

\begin{proof}
All data are in the leaves. Let each internal node $t$ contain a
number $C(t)$ for the number of leaves below it. Start from the
root. For each internal node $t$ with children $t_1,\cdots t_k$.
With probability ${C(t_i)\over C(t_1)+\cdots +C(t_k)}$, go the next
node $N_i$. Clearly, a trivial induction can show that each leaf has
an equal chance to be returned.
\end{proof}

Each set is represented by a B-tree that saves all data in the
leaves. Let each internal node contain a number for the number of
leaves below it. It takes $\bigO(\log \setSize)$ time to generate a
random element, and $\bigO(\log \setSize)$ time to check if an
element belongs to a set. We have $r(\setSize)=\bigO(\log \setSize)$
and $q(\setSize)=\bigO(\log \setSize)$.

\begin{theorem}Assume each input set is in a B-tree. Then
we have \begin{itemize}
\item
it takes $\bigO(\log \setSize)$ time for insertion and deletion, and
\item
with probability at least $1-\gamma$,
ApproximateMaximumCover($A_1,A_2,\cdots, A_{\setCount}, \epsilon,
\gamma$) outputs a $(1-{1\over e})$-approximation in time
$\bigO(T(\epsilon, \gamma, \setCount)+R(\epsilon,
\gamma,\setCount)\log \setSize +Q(\epsilon, \gamma, \setCount)\log
\setSize)$, where $T(.), R(.)$, and $Q(.)$ are the same as those in
Theorem~\ref{appr2-cover-theorem}, $\epsilon,\gamma\in (0,1)$ and
$\setSize=\max\{|A_1|,|A_2|,\cdots,|A_{\setCount}|\}$.
\end{itemize}
\end{theorem}

\begin{proof} The input list of sets is a $((0,0),(0,0))$-list as B-tree provide perfect uniform random sampling, and has
the exact number of items for each set.
 It follows from Lemma~\ref{improve-ramdom-sample-lemma} and
Theorem~\ref{main2-theorem}.
\end{proof}

\subsection{Maximum Coverage with Hashing Table}

Each set $S_i$ is saved in an unsorted  array $A_i[\ ]$. A hashing
table $H_i[\ ]$ is used to indicate if an element belongs to a set.
We can let each cell $j$ of hashing table to contain a linked list
that holds all the elements $x$ in $S_i$ with $H_i(x)=j$.

\begin{definition}
Let $\alpha(\setSize)$ and $\beta(\setSize)$ be two functions from
$\integerSet$ to $\integerSet$ . A set $S$ of $\setSize$  elements
is saved in a $(\alpha(.), \beta(.))$-hashing table
$H[1...\hashingSize]$  if the following conditions are satisfied:
\begin{itemize}
\item
There is an integer $\hashingSize\le \alpha(\setSize) \setSize$.
\item
There is a hashing function $h(.)$ with range $\{1,2,\cdots,
\hashingSize\}$ such that there are most $\beta(\setSize)$ elements
in $S$ to be mapped to the same value by the function $h(.)$. In
other words, $|x: x\in S$ and $h(x)=j\}|\le \beta(\setSize)$ for
every $j\in \{1,2,\cdots, \hashingSize\}$.
\item
The table $H[1...\hashingSize]$ is of size $\hashingSize\le
\alpha(\setSize)\setSize$ such that entry $H[j]$ points to a B-tree
that stores all the elements $x\in S$ with $h(x)=j$.
\end{itemize}
\end{definition}

Assume each input set is saved in  a $(\alpha(.), \beta(.))$-hashing
table. Each set $S_i$ is saved in an unsorted  array $A_i[\ ]$. A
hashing table $H_i[\ ]$ is used to indicate if an element belong to
a set. It takes $\bigO(1)$ time to generate a random element of
$S_i$ by accessing $A_i[\ ]$. It takes $\bigO(\log\beta(\setSize))$
time to check the membership problem by accessing the hashing table
$H_i(\ )$. This method makes it easy to add and delete elements from
the set. It takes $\bigO(\log\beta(\setSize))$ time for insertion
and deletion when the $A_i[\ ]$ and $H[\ ]$ are not full. It needs
to increase the hashing table size when it is full, and take
$\bigO(\setSize(\alpha(\setSize)+\log \beta(\setSize)))$ time to
build a new table. Thus, it takes $\bigO(\alpha(\setSize)+\log
\beta(\setSize))$ amortized time for insertion and deletion.

The array size $A_i[\ ]$ and hashing table size $H_i[\ ]$ are larger
than the size of the set $S_i$. When one of them is full, it will be
doubled by applying for a double size memory. If its usage is less
than half, it can be shrunk  by half. We show the existence of a
$(O(1),O(\log \setCount))$-Hashing Table for a set of size
$\setSize$  under some assumption.
\begin{definition}
For a hashing table $H[1...\hashingSize]$ of size $M$ and a hashing
function $h(.)$ with range $\{1,2,\cdots, \hashingSize\}$, function
$h(.)$ is $d$-uniform if $\prob(h(x)=j)\le {d\over \hashingSize}$
for every $j\in \{1,2,\cdots, \hashingSize\}$, where $d$ is a real
number in $[1,+\infty)$.
\end{definition}


\begin{theorem}
Let $\alpha(\setSize)$ and $\beta(\setSize)$ be two functions from
$\integerSet$ to $\integerSet$ . Assume each input set is saved in
a $(\alpha(.), \beta(.))$-hashing table. Then there is a
$\bigO(T(\epsilon, \gamma, k,\setCount)+R(\epsilon,
\gamma,k,\setCount)+Q(\epsilon,
\gamma,k,\setCount)\log\beta(\setSize))$ time randomized algorithm
ApproximateMaximumCover(.) such that with probability at least
$1-\gamma$, ApproximateMaximumCover($A_1,A_2,\cdots, A_{\setCount},
\epsilon, \gamma$) does not output an  $((1-{1\over e}),
\epsilon)$-approximation for the maximum coverage in type 0 model,
where $\epsilon,\gamma\in (0,1)$,  $T(.), R(.)$, and $Q_U(.)$ are
the same as those in Theorem~\ref{appr2-cover-theorem}.
\end{theorem}
\begin{proof} It takes $\bigO(1)$ time to generate a random element, and
$\bigO(\beta(\setSize))$ time to make a membership query when a set
is saved in a $(\alpha(.), \beta(.))$-hashing table. It follows from
Theorem~\ref{main2-theorem}.
\end{proof}

We tend to believe that a set can be saved in a
$(O(1),O(1))$-hashing table. Assume that each set can be saved in a
$(O(1), O(1))$-hashing table. It takes $\bigO(1)$ time to generate a
random element of $S_i$ by accessing $A_i[\ ]$. It takes $\bigO(1)$
time to check the membership problem by accessing the hashing table
$H_i(\ )$. We have $r(\setSize)=\bigO(1)$ and
$q(\setSize)=\bigO(1)$.

\begin{proposition}Let $d$ be a fixed real in $[1,+\infty)$. Let $a_1\le a_2$ be fixed real numbers in $(0,+\infty)$,
 and $\epsilon$ be fixed real numbers in $(0,1)$.
Let $h(.)$ be a $d$-uniform hashing function for some $\hashingSize$
in the range $[a_1\setSize, a_2\setSize)$ for fixed $a_1$, and
$a_2$. Then every set of size $\setSize$ has a $(O(1),
O({b(\setSize)\log \setSize\over \log\log \setSize}))$-Hashing Table
via $h(.)$ with probability at least $1-\littleo(1)$ for any
nondecreasing unbounded function $b(\setSize):$
$\integerSet\rightarrow \integerSet$.
\end{proposition}

\begin{proof}
Consider an arbitrary integer $j$ with $1\le j\le \hashingSize$ in
the hashing table. Let $S$ have the elements $x_1, x_2, \cdots,
x_{\setSize}$.
Let $\hashingSize$ be an arbitrary $a_1\setSize\le \hashingSize\le
a_2\setSize$ for two fixed $a_1,a_2\in (0,+\infty)$. The probability
that $h(x)=j$ is at most ${d\over \hashingSize}\le {d\over
a_1\setSize}=p$. Let $y(\setSize)={b(\setSize)\log \setSize\over
\log\log \setSize}$. By Theorem~\ref{ourchernoff2-theorem}, with
probability at most $P_1=({e^{y(\setSize)}\over
(1+y(\setSize))^{1+y(\setSize)}})^{p\setSize}\le
({e^{y(\setSize)}\over (1+y(\setSize))^{1+y(\setSize)}})^{d\over
a_1}=\littleo({1\over \setSize})$, we have $|\{a_i:h(a_i)=j\}|\ge
(1+y(\setSize))p\setSize=(1+y(\setSize)){d\over
a_1}=O(y(\setSize))$.

With probability at most $\hashingSize\cdot P_1=\littleo(1)$, one of
the $\hashingSize$ positions in the hashing table has more than
$cy(\setSize)$ elements of $S$ to be mapped by $h(.)$ for some fixed
$c>0$. Thus, set $S$ has a $(O(1), O(y(\setSize)))$-Hashing Table
via $h[.]$ with probability at least $1-\littleo(1)$.
\end{proof}

\begin{proposition}Let $d$ be a fixed real in $[1,+\infty)$. Let $1\le a$ be fixed real numbers in $(0,+\infty)$,
 and $\epsilon$ be fixed real numbers in $(0,1)$.
Let $h(.)$ be a $d$-uniform hashing function for some $M$ in the
range $[\setSize^{1+\epsilon}, a\setSize^{1+\epsilon}]$. Then every
set of size $\setSize$  has a $(\setSize^{\epsilon}, O(1))$-Hashing
Table via $h(.)$ with probability at least $1-\littleo(1)$.
\end{proposition}

\begin{proof}
Consider a position $j$ with $1\le j\le \hashingSize$ in the hashing
table. Let $S$ have the elements $x_1, x_2, \cdots, x_{\setSize}$.
Select a constant $c={ec_1}$ with $c_1={100\over \epsilon}$. Let
$g(\setSize)=\setSize^{\epsilon}$. Let $\hashingSize\ge \setSize
g(\setSize)$. Let $p={d\over \hashingSize}$. By
Theorem~\ref{ourchernoff2-theorem}, with probability at most
$P_1=({e^{cg(\setSize)}\over
(1+cg(\setSize))^{1+cg(\setSize)}})^{p\setSize}$,
$|\{a_i:h(a_i)=j\}|\ge (1+cg(\setSize))p\setSize$.

On the other hand, $(1+cg(\setSize))p\setSize\le
2cg(\setSize)p\setSize=2cg(\setSize)\cdot {d\over \hashingSize}\cdot
\setSize\le 2cg(\setSize)\cdot {d\over \setSize g(\setSize)}\cdot
\setSize =2cd$. Let $k=2cd=2ec_1d=2e\cdot {100\over \epsilon}\cdot
d={200d\over \epsilon}$.

We have
\begin{eqnarray}
P_1&=&({e^{cg(\setSize)}\over
(1+cg(\setSize))^{1+cg(\setSize)}})^{p\setSize}
\le({e^{cg(\setSize)}\over (1+cg(\setSize))^{1+cg(\setSize)}})^{d\setSize\over  \hashingSize}\\
&\le&({e^{cg(\setSize)}\over
(1+cg(\setSize))^{cg(\setSize)}})^{d\setSize\over  \hashingSize}
\le({e^{cg(\setSize)}\over (cg(\setSize))^{cg(\setSize)}})^{d\setSize\over  \hashingSize}\\
&\le&({e^{cg(\setSize)}\over
(cg(\setSize))^{cg(\setSize)}})^{d\setSize\over \setSize
g(\setSize)} \le({1\over c_1g(\setSize)})^{cd} \le({1\over
c_1\setSize^{\epsilon}})^{cd} \le({1\over
c_1\setSize})^{cd\epsilon}.
\end{eqnarray}

With probability at most $\hashingSize P_1\le 2g(\setSize)\setSize
({1\over c_1\setSize})^{cd\epsilon}=\littleo(1)$, one of the
$\hashingSize$ positions in the hashing table has more than $k$
elements of $S$ to be mapped by $h(.)$.

\end{proof}


\ifdefined\RegularVersion

\section{Conclusions}

We developed a  randomized greed approach for the maximum coverage
problem. It obtains the same approximation ratio $(1-{1\over e})$ as
the classical approximation for the maximum coverage problem,  while
its computational time is independent of the cardinalities of input
sets under the model that each set answers query and generates one
random sample in $\bigO(1)$ time. It can be applied to find
approximate maximum volume by selecting $k$ objects among a list of
objects such as rectangles in high dimensional space. It can provide
an efficient online implementation if each set is saved in a B-tree.
Our approximation ratio depends on the how much the random sampling
is biased, and the initial approximation accuracy for the input set
sizes. The two accuracies are determined by the parameters
$\alpha_L, \alpha_R,\delta_L$ and $\delta_R$ in a
$((\alpha_L,\alpha_R),(\delta_L,\delta_R))$-list. It seems that our
method can be generalized to deal with more general version of the
maximum coverage problems, and it is expected to obtain more results
in this direction. The notion of partial sublinear time algorithm
will be used to characterize more computational problems than the
sublinear time algorithm.

\section{Acknowledgements}
The author is grateful to Jack Snoeyink for his helpful comments
that improve the presentation of this paper.


\fi

\end{document}